\documentclass[12pt]{article}
\usepackage[hyphens]{url}
\usepackage{amsfonts}
\usepackage{eurosym}
\usepackage{graphicx}
\usepackage{booktabs}
\usepackage{float}
\usepackage{color}
\usepackage{soul}
\usepackage[usenames,dvipsnames]{xcolor}
\usepackage{tikz}
\usetikzlibrary{shapes,snakes}
\usepackage{subfig}
\usepackage{amssymb}
\usepackage[square,numbers]{natbib}
\usepackage{comment}
\usepackage{amsmath}
\usepackage{amsthm}
\usepackage{booktabs}
\usepackage{multirow}
\usepackage{enumerate}
\usepackage{color}
\usepackage{hyperref}
\definecolor{MyDarkBlue}{rgb}{0,0.08,0.45}
\usepackage[para]{threeparttable}
\usepackage{verbatim}
\usepackage{fancyvrb}
\usepackage{pdflscape}
\usepackage{lineno}
\usepackage[font=small,labelfont=bf]{caption}
\usepackage{multibib}

\newcites{New}{ Supplemental References}

\hypersetup{colorlinks=false, linkcolor=red, citecolor=blue}

\newtheorem{theorem}{{\bf \sc Theorem}}

\renewcommand{\Pr}{\mathbb{P}}
\begin{document}

\title{Interacting Regional Policies in Containing a Disease}

\author{Arun G. Chandrasekhar\thanks{Department of Economics, Stanford University; J-PAL; NBER.}, Paul Goldsmith-Pinkham\thanks{Yale School of Management.},\\   Matthew O. Jackson\thanks{	Department of
		Economics, Stanford University; Santa Fe Institute.}, and Samuel Thau\thanks{Harvard University.
		\newline
		 We thank Abhijit Banerjee, Gabriel Carroll, Bharat Chandar, Dean Eckles, Ben Golub, Dave Holtz, Ali Jadbabaie, Ed Kaplan,  and  Johan Ugander  for helpful discussions.
		 We  gratefully acknowledge financial support from the NSF under grant SES-1629446 and
		 RAPID \# 2029880.  We thank J-PAL SA, Tithee Mukhopadhyay, Shreya Chaturvedi,
		 Vasu Chaudhary, Shobitha Cherian, Arnesh Chowdhury,   Anoop Singh Rawat, and Meghna Yadav for research assistance. The computations in this paper were run on the FASRC Cannon cluster supported by the FAS Division of Science Research Computing Group at Harvard University.}}
\date{January 2021}

\maketitle
\begin{abstract}
Regional quarantine policies, in which a portion of a population surrounding infections are locked down,
are an important tool to contain disease.  However, jurisdictional governments -- such as cities,
counties, states, and countries -- act with minimal coordination across borders.
We show that a regional
quarantine policy's effectiveness depends upon whether (i) the network of interactions satisfies a
balanced-growth condition, (ii) infections have a short
delay in detection, and (iii) the government has control over and knowledge of the necessary
parts of the network (no leakage of behaviors).
As these conditions generally fail to be satisfied, especially when interactions cross borders,
we show that substantial improvements are possible if governments are outward-looking and proactive:
triggering quarantines in reaction to neighbors' infection rates, in some cases even before infections are detected internally.
We also show that even a few lax governments -- those that wait for nontrivial internal infection rates before quarantining -- impose substantial costs on the whole system.  Our results
illustrate the importance of understanding contagion across policy borders
and offer a starting point in designing proactive policies for decentralized jurisdictions.
\end{abstract}


\thispagestyle{empty}

\setcounter{page}{0} \newpage

\section*{Introduction}
Global problems, from climate change to financial crises to disease control, are hard to address without policy coordination across borders.   Carbon emissions in one region are everyone's problem, as are financial collapses, as well as the spread of an infectious disease. 
Coordinating policies across jurisdictions both in terms of timing and scale are important whenever problems have spillovers.  In this paper we shed light on this problem
by examining how different types of decentralized policies fare compared to
more centralized policies at containing the spread of an infectious disease.

In particular,
pandemics, like COVID-19, are challenging to contain if
governments fail to coordinate efforts. Without vaccines or herd immunity,
governments have responded to infections by limiting constituents'
interactions in areas where an outbreak exceeds a threshold of
infections.  Such regional quarantine policies are used by towns,
cities, counties, states, and countries, and trace to the days of the
black plague.  Over the past 150 years, regional quarantines have been
used to combat cholera, diphtheria, typhoid, flus, polio, ebola, and
COVID-19
\cite{hardy1993cholera,gensini2004concept,tognotti2013lessons,drazen2014ebola},
but rarely with coordination across borders.

Decentralized policies across jurisdictions have two major
shortcomings.  First, governments care primarily about their own
citizens and do not account for how their infections impact other
jurisdictions: the resulting lack of coordination can lead to worse
overall outcomes than a global policy
\cite{jacksonl2013,holtz2020interdependence,cheng2020variations}.
Second, many governments are inward-looking,  only paying attention to internal situations, which leads them to under-forecast their own infection
rates.

We examine three types of quarantine policies to understand the impact
of non-coordination: (i) those controlled by one actor with control of
the whole society -- ``\emph{single regime policies},'' (ii) those controlled
by separate jurisdictions that are inward-looking and only react to internal infection rates, or ``{\sl reactive}'' for short and (iii) those controlled by separate jurisdictions that are outward-looking, tracking infections outside of their jurisdiction as well as within to forecast their infection rates when deciding when to quarantine, or ``{\sl proactive}'' for short.

We use a general model of contagion through a network to study these
policies.  We first consider single regime policies.  A government
can quarantine everyone at once under a ``global quarantine,'' but
those are very costly (e.g., lost days of work, school, etc.).  Less costly (in the short run),
and hence more common, alternatives are
``regional quarantines'' in which only people within some distance of
observed infections are quarantined.  Regional quarantines, however, face two
challenges.  First, many diseases are difficult to detect,
because some individuals are either asymptomatically contagious (e.g., HIV,
COVID-19)
\cite{bridges2003transmission,ten2018contributions,bai2020presumed}, 
or a government lacks resources to quickly identify
infections \cite{wu2020characteristics,lazer2020stateofnation}.
Second, it may be infeasible to fully quarantine a part of
the network, because of difficulties in identifying whom to
quarantine  (e.g., imperfect or inefficient contact tracing)
or non-compliance by some people -- by choice or necessity
\cite{parmet2020covid,wilder2020can,ghani1997role,jolly2002gonorrhoea,halloran2002containing,keeling2005networks}
. Either way, tiny leakages can spread the disease.

We show that regional quarantines curb the spread of a disease if and
only if: (i) there is limited delay in observing infections, (ii)
there is sufficient knowledge and control of the network to prevent
leakage of infection, and (iii) the network has a certain
``balanced-growth'' structure.  The failure of any of these conditions
substantially limits regional quarantine effectiveness.

We then examine jurisdictional policies, which are regional quarantine
policies conducted by multiple, uncoordinated regimes.
The regions that need to be quarantined, however, often cross borders, leading to
leakage that limits their
effectiveness.  As we show, jurisdictional policies that are reactive do much worse than
proactive ones, as they do not forecast the impact of neighboring jurisdictions'
infection rates on their own population.  Moreover, a few lax
jurisdictions, which wait for higher infection rates before
quarantining, worsen outcomes for all jurisdictions.

\section*{A Model}

Consider a large network of $n$ individuals or nodes.  Our theory is asymptotic, stating properties that apply with a probability approaching one as the population grows ($n\rightarrow \infty$). Our theoretical results consider sequences of networks as $n$ grows, while the simulations are on given networks with thousands of nodes. 

An infectious disease begins with an infection of a
node $i_0$, the location of which is known, and expands via (directed) paths from $i_0$. In each discrete time period, the infection spreads from each currently infected node to each of its susceptible contacts
independently with probability $p$. A node is infectious for $\theta$ periods, after which it
recovers and is no longer susceptible or contagious, though our results extend to the case in
which a node can become susceptible again.

The disease may exhibit a delay of $\tau \leq \theta$ periods during which an infected and contagious person
does not test positive. This can be a period of asymptomatic infectiousness, a delay in testing,
or limits to healthcare access \cite{bridges2003transmission,bai2020presumed,lauer2020incubation,wu2020characteristics,lazer2020stateofnation,banerjee2020messages}. 
After that delay, each infected node's infection is detected
with probability $\alpha<1$ (for simplicity, in the first period after the delay).
$\alpha$ incorporates testing accuracy, availability, and decisions to test.

This framework nests the susceptible-infected-recovered (SIR) model
and its variations including exposure, multiple infectious stages, and death \cite{kermack1927contribution,bailey1957mathematical,anderson1992infectious,keeling2005networks,mcadams2020}, agent-based models \cite{newman1999scaling,newman2002spread,flaxman2020estimating}, and others.

\section*{Results}
\subsection*{Baseline: A Single Jurisdiction with Complete Control}

We begin by analyzing a single jurisdiction with complete control, or equivalently, the entire network being one jurisdiction with one policymaker.

A $(k,x)$-regional policy is triggered once
$x$ or more infections are observed within distance $k$ from the seed node $i_0$, at which point it
quarantines all nodes within distance $k+1$ of the seed
for $\theta$ periods.
This captures a commonly used policy where regions that are exposed to the disease are
shut down in response to detection.
We begin by giving the policymaker the advantage of knowing which nodes are within distance $k+1$ of the seed,
which could reflect rapid and efficient contact tracing supplemented with rich network data.
We later explore how errors in this knowledge change the results.
We also give the policymaker knowledge of which
node is the seed and study subsequent containment efforts.
In practice, policymakers must estimate the origin of infection, which presents an additional challenge.

Whether a regional policy halts infection is fully characterized by whether the sequence of networks satisfies what we call \emph{growth-balance}. This requires that as the networks become large, all paths along which the disease might escape beyond the regional quarantine of distance $k+1$ are such that at least some nodes have many neighbors. In particular, for the sequence of networks to be growth-balanced with respect to $k$, there must exist $m(n)\rightarrow \infty$ such that in a network with $n$ nodes, every path leading from $i_0$ to a node at distance $k+2$ has at least one node with degree at least $m(n)$.  This condition ensures that if an infection does spread along some path that can take it outside of the region then, at some point along the way, it is very likely to infect many nodes and thus be detected before it reaches the edge of the region.

To better understand growth-balance, consider an example of a disease that is
beginning to spread with a reproduction number $R_0$ of 3.5 and
such that one in ten cases are detected in a timely manner $(\alpha = 0.1)$.
First, consider a part of the network in which
each infected person infects 3.5 others on average.
If we monitor all nodes within distance $k=3$ of an infected node,
a ``typical'' path of infection would lead to roughly $3.5+3.5^2+3.5^3= 58.625$ expected cases before it reaches the edge of the region.
The chance that this goes undetected is tiny: $0.9^{58.625} = 0.002$.
In contrast, suppose the infection starts in a part of the network where each infected person infects
just one other, on average, so that the local reproduction number here is $R_0 = 1$ rather than 3.5.
Now a path of length 3 leads to $1+1+1=3$ (expected) infections.
The chance that such a spread remains undetected is very high: $0.9^3=0.72$.

Speaking loosely,  many different networks can lead to the same average reproduction number, but have very different structures.
If the distribution of reproduction numbers around the network has no pockets in which they are too low -- i.e., if the growth structure of the disease around the network is well ``balanced'' and not too low --  then
it is highly likely that any early infection will be detected before it gets too far from the first infected node.
If instead, the distribution of reproduction numbers gives a
nontrivial chance that the disease starts out on a path with all low reproduction numbers,
like the 1, 1, 1, path,  then there is a high chance that it can travel far from the seed
before being detected.
This highlights the fact that a reproduction number $R_0$ alone is a crude concept,
and that the specifics of the network structure matter considerably for whether a disease spreads or is containable.
In particular, areas with low $R_0$ (but above one) can lead to more containment failures and
lead to broader infections.
Given the short distances in many networks \cite{watts1998collective,amaral2000classes,chung2002average},
a lack of growth balance allows a disease to spread far before detection.
Figure S1 in the supplementary materials pictures a network
that has a high average reproduction number,
but is not growth-balanced and allows the infection to travel far from the initially infected node without detection.

In Theorem 1 in the Supplementary Information we prove that,  with no delays in detection and no leakage,
a $(k,x)$-regional policy halts infection among all nodes beyond distance $k+1$ from $i_0$
with probability approaching 1 (as the population grows) {\sl if and only if} the sequence of networks
satisfies growth-balance with respect to $k$.   In fact, we prove a stronger version in which the quarantine distance $k(n)$, average degree $d(n)$, transmission probability $p(n)$, delay $\theta(n)$, and detection probability $\alpha(n)$ are all allowed to vary with $n$.

Growth-balance is satisfied by some, but not all, sequences of prominent random graph models, provided that the average degrees $d(n)$ satisfy $d(n)^{k(n)}\rightarrow\infty$ (
see SI).
However, the additional heterogeneity in human contact networks makes the property unlikely to hold in real networks even if average connectivity is high.
Indeed, if contact networks have some low degree nodes (as they tend to empirically),
then, unless quarantine
regions are large ($k(n)$ grows sufficiently as $n$ grows), growth-balance fails and a regional quarantine will be ineffective at halting a spread.

Next, we show that the effectiveness of a regional policy breaks down, even if a network is growth-balanced, once there is sufficient delay in detection or leakage (due to imperfect information, enforcement, or jurisdictional boundaries).

\subsection*{Delays in Detection and Wider Quarantines}

To understand how delays in detection affect a regional policy,
consider two extremes. If the delay is short relative to the
infectious period, the policymaker can still anticipate the disease
and adjust by simply enlarging the area of the quarantine to include a
buffer.
An easy extension of the above theorem is that a regional
policy with a buffer works if and only if the sequence of networks is
growth-balanced and the delay in detection plus $k+1$ is shorter than the
diameter of the network (Supplementary Information, Theorem 2).
Given that real-world networks have short average distances between
nodes \cite{watts2004small}, this condition can be even harder to satisfy and
non-trivial delays in detection
allow the disease to escape a regional quarantine.

\subsection*{Leakage}

Next, we consider how leakage  --
the inability to limit interactions \cite{parmet2020covid} or mistakes in identifying the portion of the network and nodes to quarantine \cite{halloran2002containing,keeling2005networks} -- diminishes the effectiveness of regional policies.
Although minimizing leakage increases the chance that a regional quarantine will be successful,
we show that even a small amount of leakage leads to a nontrivial
probability that the disease will escape the quarantine.  In particular,
we show that if even a small share $\varepsilon>0$ of nodes within distance $k$ of the seed $i_0$ ignore the quarantine and are connected to nodes outside of the quarantine, then the policy will fail to halt the spread with a probability bounded away from 0 as $n$ grows (Supplementary Information, part 2 of Theorem 3).
The result also highlights a tension in containment strategies:
the infection is easier to detect when there are many nodes and interactions within the potential
quarantine radius; but there is also more leakage
and a higher chance that the infection escapes the quarantine.

\subsection*{Jurisdictions and Leakage}

We use the theory results as a starting point to understand jurisdictional policies. It is important to note that the results on leakage (Theorem 3) apply when interactions cross jurisdictions. To show this, Figure \ref{fig:jurisdiction_example} displays two jurisdictions that fail to nicely tessellate the network: geographic location and network distance from the seed are not perfectly aligned. Therefore, a quarantine in one region will necessarily have leakage, missing nodes that interact across jurisdictions. Given leakage across borders, unless policies are coordinated across jurisdictions, our theoretical results indicate that they will fail to contain infections, which is then the starting motivation of the simulations. 

\begin{figure*}
	\centering
	\textbf{Figure \ref{fig:jurisdiction_example}: Inconsistency of Jurisdictions and Distances}\par\medskip
	\subfloat[Jurisdictions with interactions that do not align\label{subfig-1:example}]
	{%
		\includegraphics[width=0.48\linewidth]{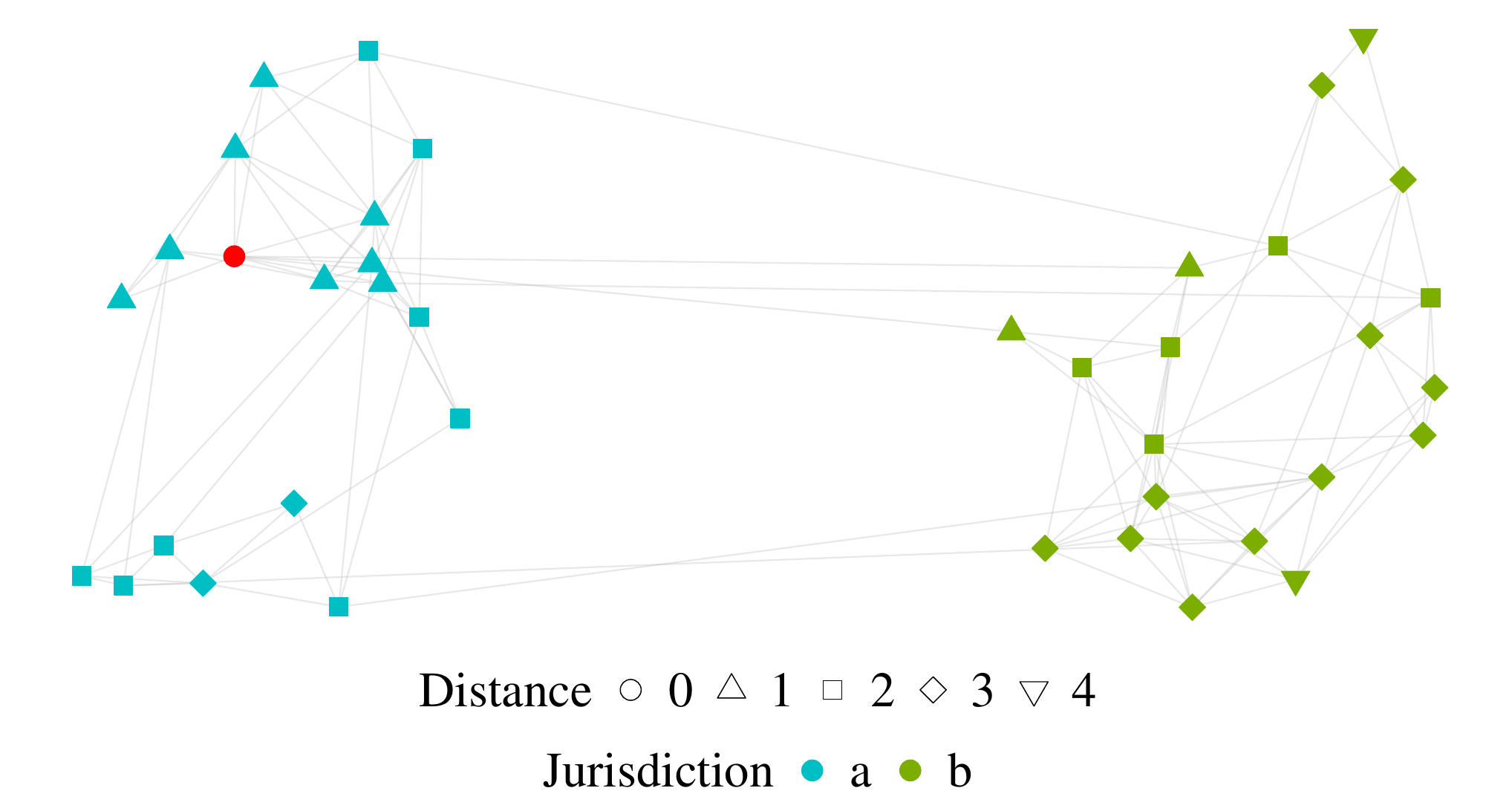}
	}
	\medskip
	\hfill
	\subfloat[Figure \ref{subfig-1:example}, but based on distance from infection \label{subfig-2:example}]{%
		\includegraphics[width=0.48\linewidth]{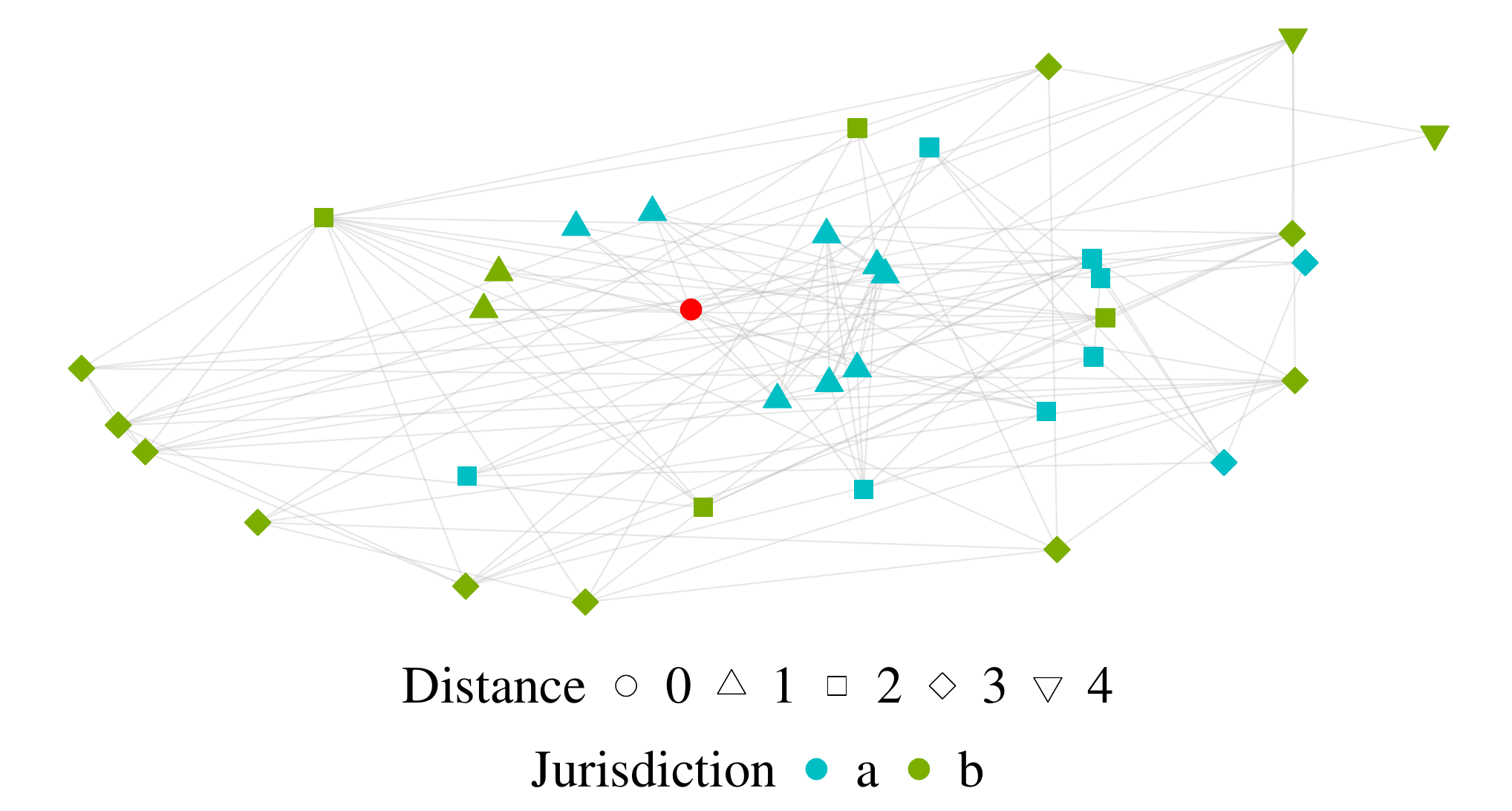}
	}
	\caption{ Nodes in two jurisdictions do not align with the distances from the initial infection. In Panel (a), the nodes are presented in a geographic sense, within their jurisdictions, and the interaction network does not comply with the jurisdictional boundaries. In Panel B, we show the network as a function of directed distance from the initial infection.  A coordinated quarantine of distance 2 over the network in Panel (b) could contain the infection; however, if it is only executed by the infected node’s jurisdiction in panel (a) then it would fail for cross-jurisdictional connections.}
	\label{fig:jurisdiction_example}
\end{figure*}

\subsection*{Simulations}

The theory shows that for most natural settings, anything short of a global quarantine is unlikely to contain the disease.  Thus, it becomes important to understand how well different policies do at curbing the number of infections over time.  In particular, we next study --- via simulations --- how stylized versions of the containment policies that are used in practice fare in terms of minimizing infections over time, {\sl and} at what costs in terms of person-days of quarantine. 

To explore this, we simulate a contagion on a network of 140000 nodes that mimics real-world
data \cite{mccormick2010many,banerjee2018changes,beaman2018can,banerjee2020messages}.
These simulations illustrate our theoretical results and also show the improvements that
proactive policies provide relative to reactive ones.  The SI presents the robustness of the results to some variations of parameters.

The network is divided into 40 \emph{locations}, each with a population of 3500.
We generate the network using a geographic stochastic block model. The probability of interacting
declines with distance. The average degree is 20.49 and nodes have 79.08\% of their
interactions within their own locations and 20.92\% outside of their location (calibrated to data
from India and the United States, including data collected during COVID-19
\cite{mccormick2010many,banerjee2018changes,beaman2018can,banerjee2020messages, SI}). We fix this network, and use it for all simulations. 

 We conduct 10000 simulations of each policy and then take the average over the simulations, with each simulation using an infection seed selected uniformly at random. The simulations progress in four stages: first, any node that has been infected for exactly $\tau$ periods is detected with probability $\alpha$; next, policy makers use the information they have about detected infections to decide whether to enact a quarantines (if one is not already in place in their jurisdiction); third, the disease progresses and currently infected people can infect their neighbors and people who have been infected for $\theta$ periods recover; finally, quarantines can end and new quarantines are implemented. We set the rate that a node infects its neighbors in order to get a basic reproduction rate of $R_0 = 3.5$ (to mimic COVID-19 \cite{hao2020reconstruction}),
and we set $\theta = 5$, $\tau=3$ (when used) and $\alpha = 0.1$
\cite{lauer2020incubation,hortaccsu2020estimating,li2020substantial, SI}.

The simulated network is fairly symmetric in degree and therefore approximates satisfying growth-balance. Thus, the attention in our simulations is focused on leakage across jurisdictions and detection delay.

Before introducing jurisdictions, we first illustrate the effects of leakage as well as delays in detection on a regional policy.
In Figure \ref{fig:theorem_conditions}, the entire network is governed by a single policymaker using a
$(k,x)=(3,1)$-regional quarantine. The policy maker eventually knows the location of the infection seed, so that it can properly center the quarantine, but does not detect the initial infection. This is meant to emulate the difficulties of finding an initial infection in real time, but we give the policy maker the advantage of being able to trace back to the epicenter and center the quarantine once they decide to enact a quarantine. As an addition to the policy, if the initial $(k,x)$ quarantine fails to contain the disease, the policy maker treats detected infected people outside of quarantine as new seeds and quarantines all nodes within distance $k+1$ of them. In the SI, we include simulations that relax the assumption that the policy maker knows the location of the original seed $i_0$, along with variants of the transmission and detection parameters -- varying $\theta$, $\tau$, $\alpha$, and $R_0$. 

Figure \ref{subfig-1:theorem} shows the outcomes for no delay in
detection nor any leakage.  Consistent with Theorem 1
the policy is effective: on average 277 people per million are infected
(0.028\% of the population), with 803956 person-days of quarantine per million people.
Figure \ref{subfig-2:theorem} introduces a delay in
detection.
With a delay of $\tau=3$, infections increase, with 2256 people per million
eventually infected (0.23\% of the population) and 2301414 person-days of quarantine per million
people. Adding a buffer to correspond to the detection delay
effectively makes the regional policy global, as the buffered region
contains 99.98\% of the population on average.  Figure \ref{subfig-4:theorem}
adds leakage to the setup of Figure \ref{subfig-2:theorem}, by having
5\% of people never quarantine.
The number of cumulative infections per million people increases to
5138 (0.50\% of the population).  The leakage increases the
number of quarantined person-days to 6478055 per million nodes.

\begin{figure*}
	\centering
	\textbf{Figure \ref{fig:theorem_conditions}:  The Impact of Detection Delay and Leakage}\par\medskip
	\subfloat[$(k,x)=(3,1)$-quarantine with no delay in detection and no leakage \label{subfig-1:theorem}]{%
		\includegraphics[width=0.48\linewidth]{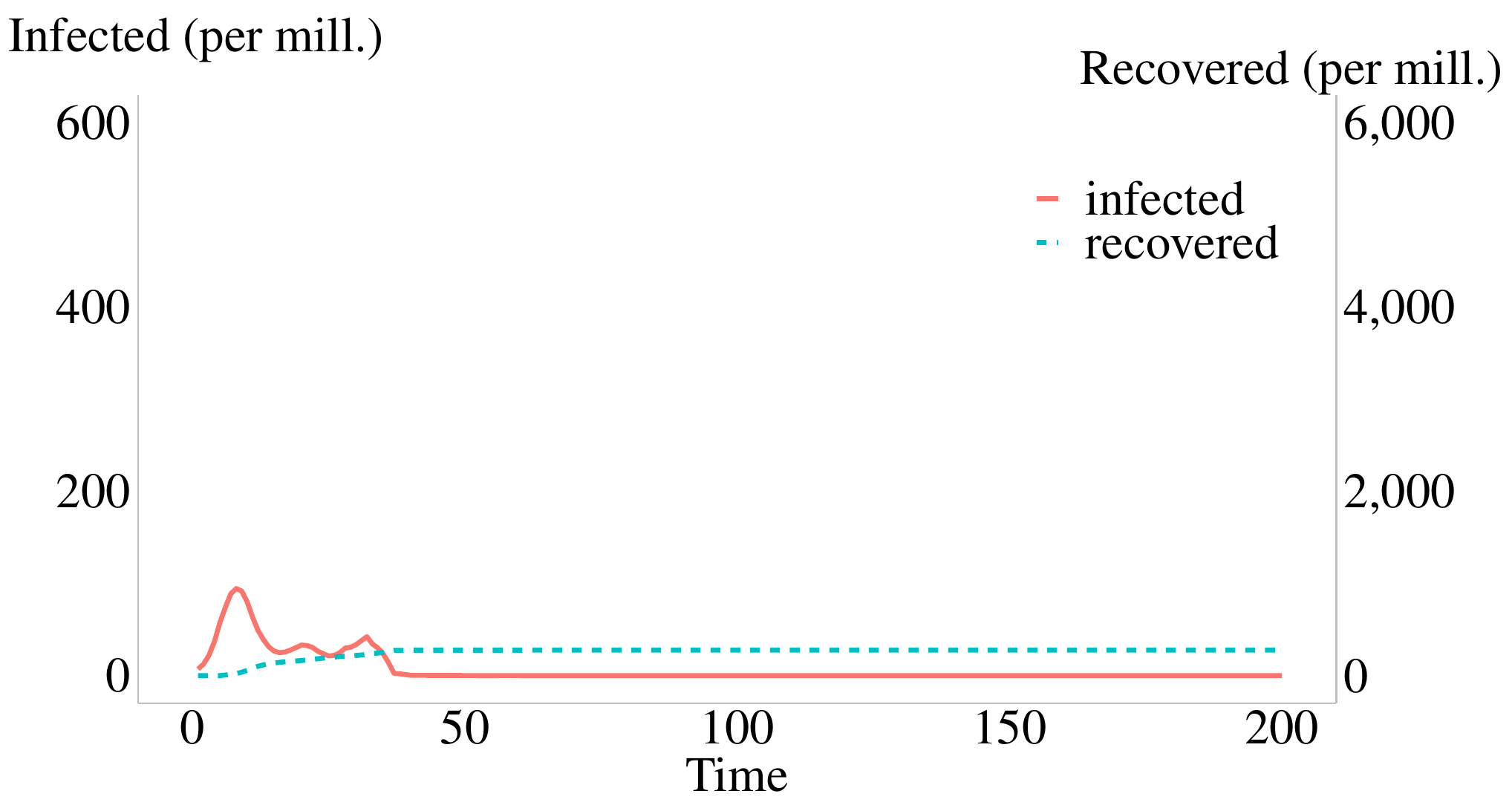}
	}\\
	\subfloat[$(k,x)=(3,1)$-quarantine policy with a detection delay of 3 periods\label{subfig-2:theorem}]{%
		\includegraphics[width=0.48\linewidth]{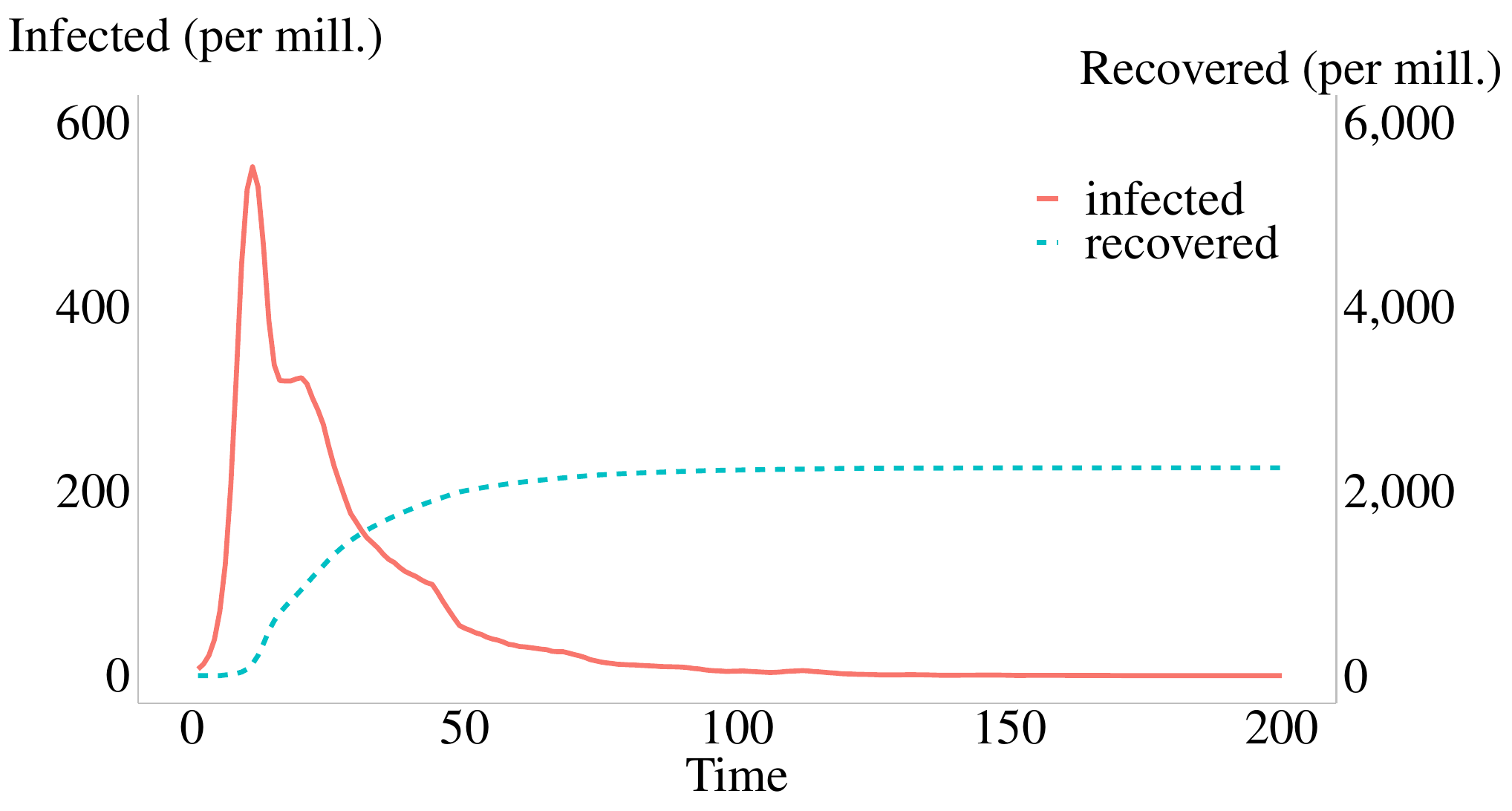}
	}
	\hfill
	\subfloat[$(k,x)=(3,1)$-quarantine policy with a detection delay of 3 periods and leakage\label{subfig-4:theorem}]{%
		\includegraphics[width=0.48\linewidth]{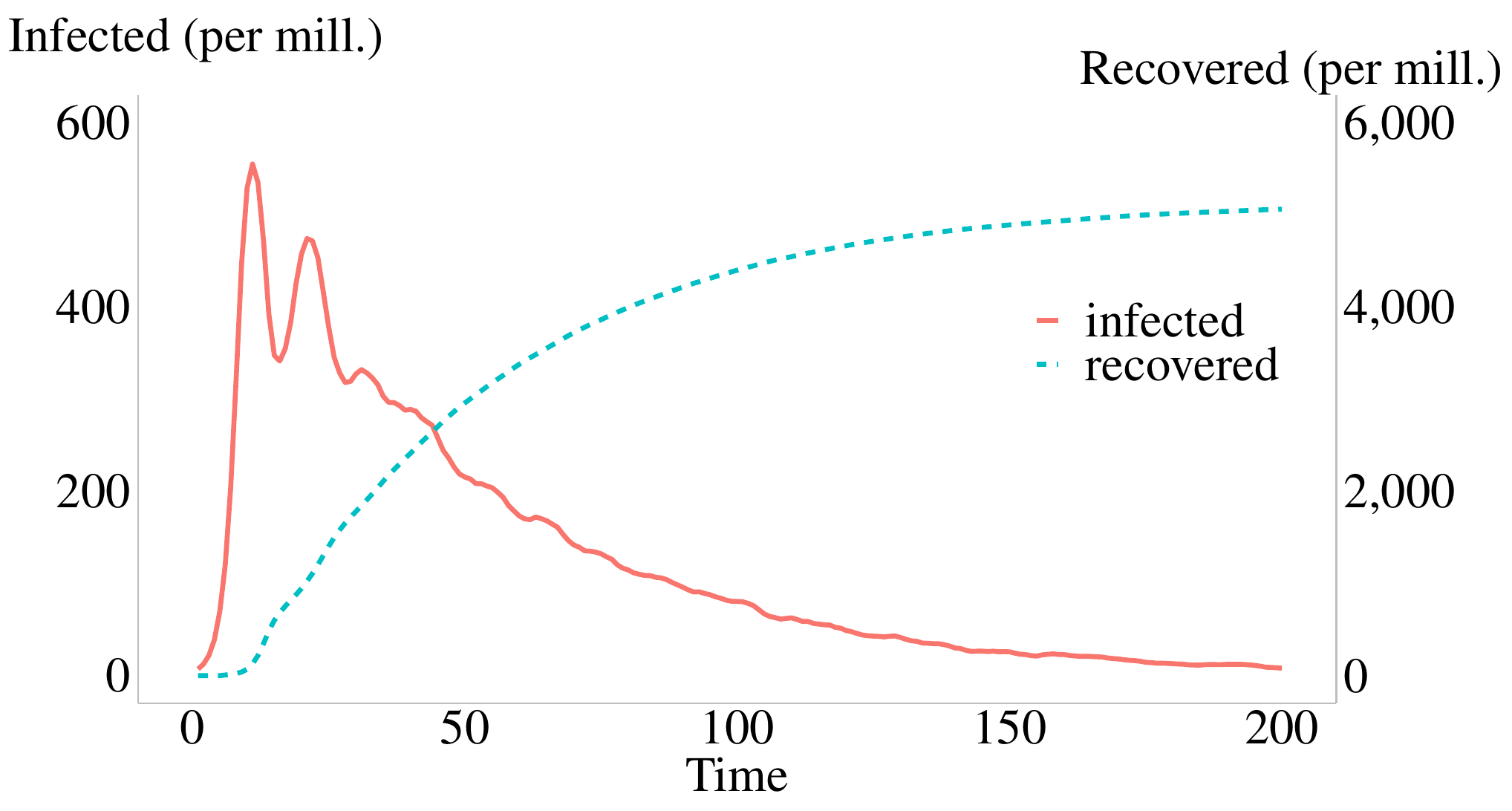}
	}
	\caption{
		We picture daily infections and cumulative recoveries under three scenarios.
		The entire network is governed by a single policymaker using a
		$(k,x)=(3,1)$-regional quarantine.  In Panel \ref{subfig-1:theorem}, there is no detection delay and no leakage.
		In Panel \ref{subfig-2:theorem}, we introduce a detection delay of $\tau = 3$.
		This represents the 3 day pre-symptomatic window during which an infected node can transmit,
		as well as an expected delay in seeking healthcare and testing upon symptom
		onset \cite{lauer2020incubation, SI}. Panel \ref{subfig-4:theorem} adds leakage
		to the setup of Panel \ref{subfig-2:theorem}, by having a randomly selected fraction $\epsilon=0.05$ never quarantine. For each figure, we simulate 10000 times on the same network
		with random initial infections, and
		present the average number of infections and recovered people over time, scaled per million. }
	\label{fig:theorem_conditions}
\end{figure*}

\subsection*{Jurisdictional Policies}

We now introduce jurisdictions to the same network as before, and each of the 40 locations in the network becomes its own jurisdiction.

We compare two types of jurisdictional policies.
In reactive policies, each jurisdiction decides on when it quarantines based entirely on internal infections. If a jurisdiction has a threshold of $x$ cases, and observes at least $x$ cases within their jurisdiction, the jurisdiction goes into quarantine.
In proactive policies, jurisdictions track infections in other jurisdictions and
predict their own  --  possibly undetected -- infections and base their
quarantines off of predicted infections. Intuitively, jurisdictions are constantly estimating infection rates (including undetected cases) in each jurisdiction based on the history of observed infections, using knowledge of the interaction rates within and across borders and the infection and latency properties of the disease. More specifically, at each time step $t$, the number of estimated infections at time $t$ is the sum of the estimated infections at $t-1$, plus the expected number of new estimated infections minus the number of expected recoveries. The expected number of new infections is calculated using the connection rates to non-quarantined jurisdictions (including own jurisdiction) and the estimated infections in those jurisdictions at $t-1$.  If at any point it is clear that the actual number of detected infections in some jurisdiction is above the estimated rate, then the estimation is updated. All jurisdictions begin by estimating that there are no infections, until at least one infection is observed. Details of this calculation are in the Supplementary Information. Proactive jurisdictions quarantine if they infer (or observe) at least $x$ cases. 

We set $x=1$ for both the reactive and proactive simulations unless otherwise specified. For both reactive and proactive jurisdiction, when a jurisdiction enters quarantine, all connections to and within the jurisdiction are severed for $\theta$ periods. As before, we set $R_0=3.5$, $\theta=5$, $\tau=3$, and $\alpha =0.1$. The policy maker does not detect $i_0$ for both the reactive and proactive policies, but does know its location when setting the quarantine. 

Figure \ref{fig:naive_vs_proactive} illustrates the improvement that proactive jurisdictional policies offer relative to reactive
jurisdictional policies. In Figure \ref{subfig-1:adaptive}, jurisdictions use reactive policies,
while in Figure \ref{subfig-2:adaptive} jurisdictions use proactive policies. In the reactive case,
there are 298911 infections per million people (28.89\% of the population),
with 131303638 person-days of quarantine
per million nodes. 
Proactive quarantining dramatically improves outcomes (Figure \ref{subfig-2:adaptive}): only
17105 people per million are infected (1.71\% of the population),
with 51328755 person-days of quarantine
per million people.  

\subsection*{Lax Jurisdictions}

Finally, we also add four ``lax'' jurisdictions to the setting.
These are jurisdictions that are reactive and have a high threshold
of internal infections before quarantining, using a threshold of $x=5$.
We examine how these few lax jurisdictions worsen the outcomes for all jurisdictions.
Figure \ref{subfig-3:adaptive} shows the outcomes when the remaining 36 jurisdictions use
reactive policies, while in Figure \ref{subfig-4:adaptive}
the remaining 36 jurisdictions using proactive policies.
Comparing Figure \ref{subfig-1:adaptive} to
\ref{subfig-3:adaptive}, infections are worse under the reactive
policies. There are 340587 people per million 
infected (34.1\% of the population), compared to 298911 (29.9\%) without the lax jurisdictions. 
Comparing this change to Figures
\ref{subfig-2:adaptive} and \ref{subfig-4:adaptive} shows that things
deteriorate {\sl relatively} more for the proactive jurisdictional policies.  The 91887
total infections per million people (9.18\% of the population) is a larger increase from 17105 (1.17\%) without lax jurisdictions.  Nonetheless, even with lax jurisdictions, the proactive policies fare better than the reactive policies (even if those do not have lax jurisdictions).


\begin{figure*}
	\centering
	\textbf{Figure \ref{fig:naive_vs_proactive}: The Effectiveness of Reactive vs Proactive Quarantines}\par
	\subfloat[Each jurisdiction quarantines once it observes any infections internally, ignores other jurisdictions \label{subfig-1:adaptive}]{%
		\includegraphics[width=0.4\linewidth]{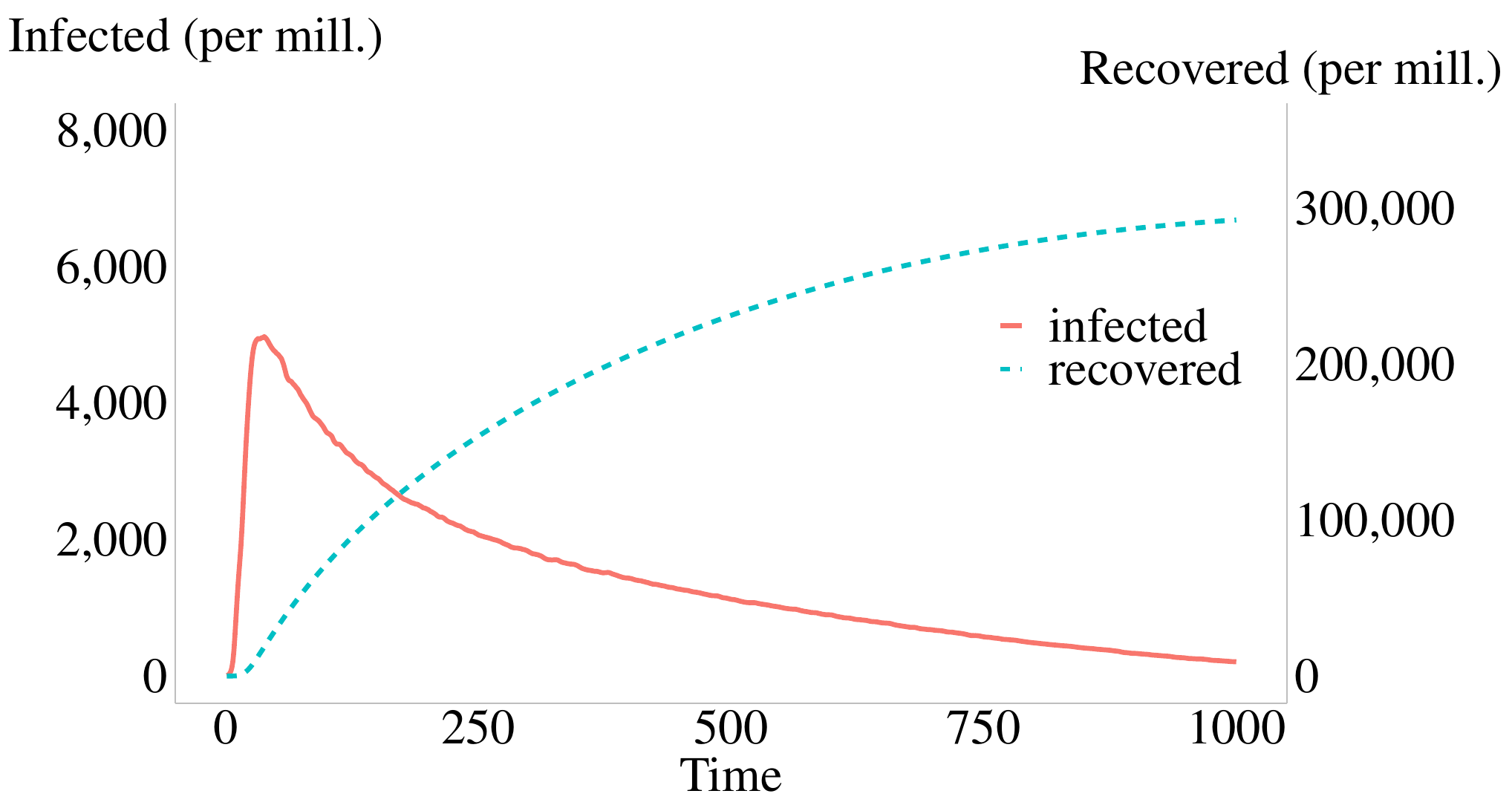}
	}
	\hfill
	\subfloat[Each jurisdiction proactively quarantines by estimating internal infections based on observation of
	other jurisdictions\label{subfig-2:adaptive}]{%
		\includegraphics[width=0.4\linewidth]{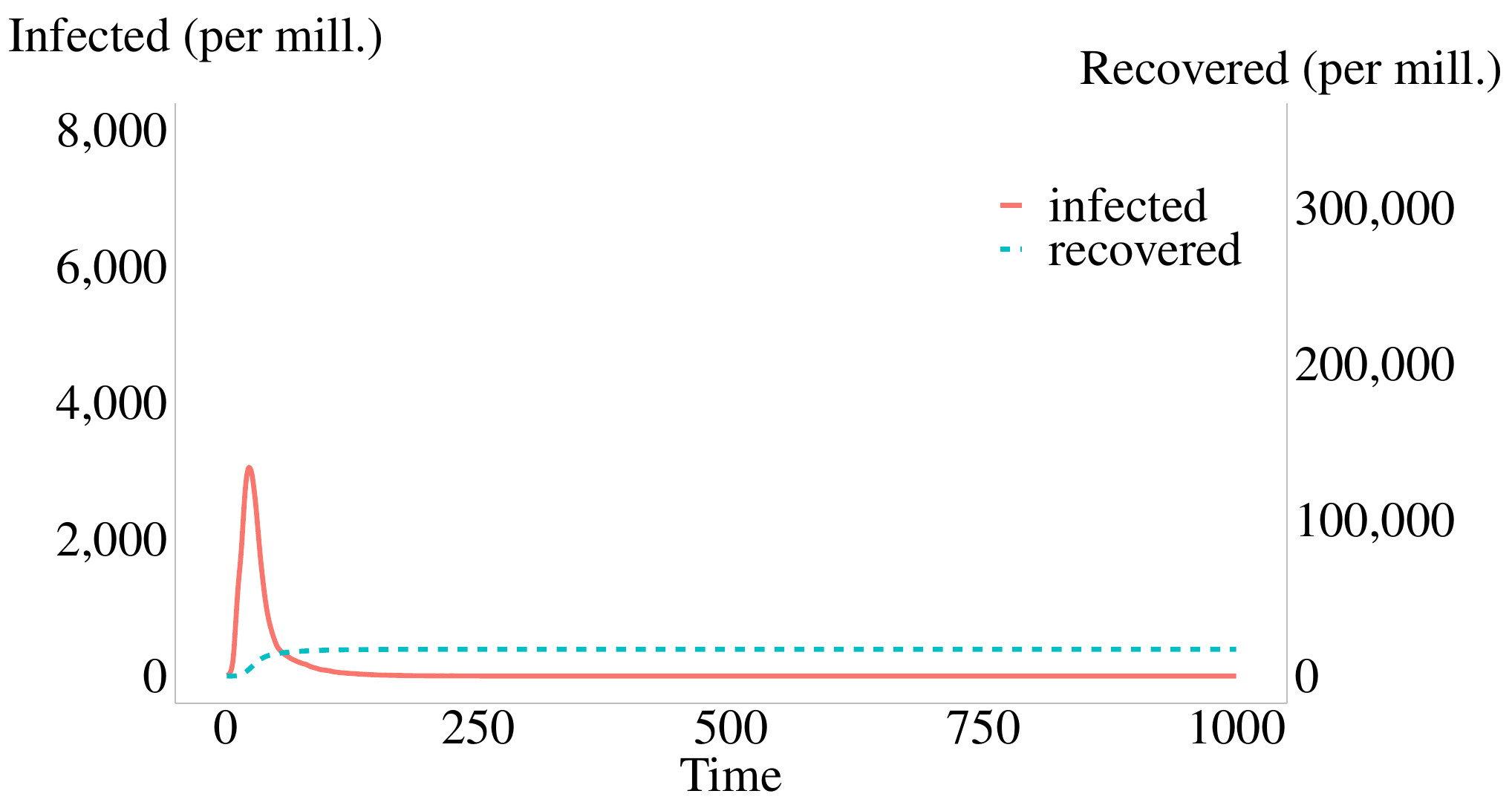}
	}\\
	\subfloat[36 jurisdictions quarantine once observing any internal infections, 4 lax jurisdictions only quarantine once they reach 5 internal infections\label{subfig-3:adaptive}]{%
		\includegraphics[width=0.4\linewidth]{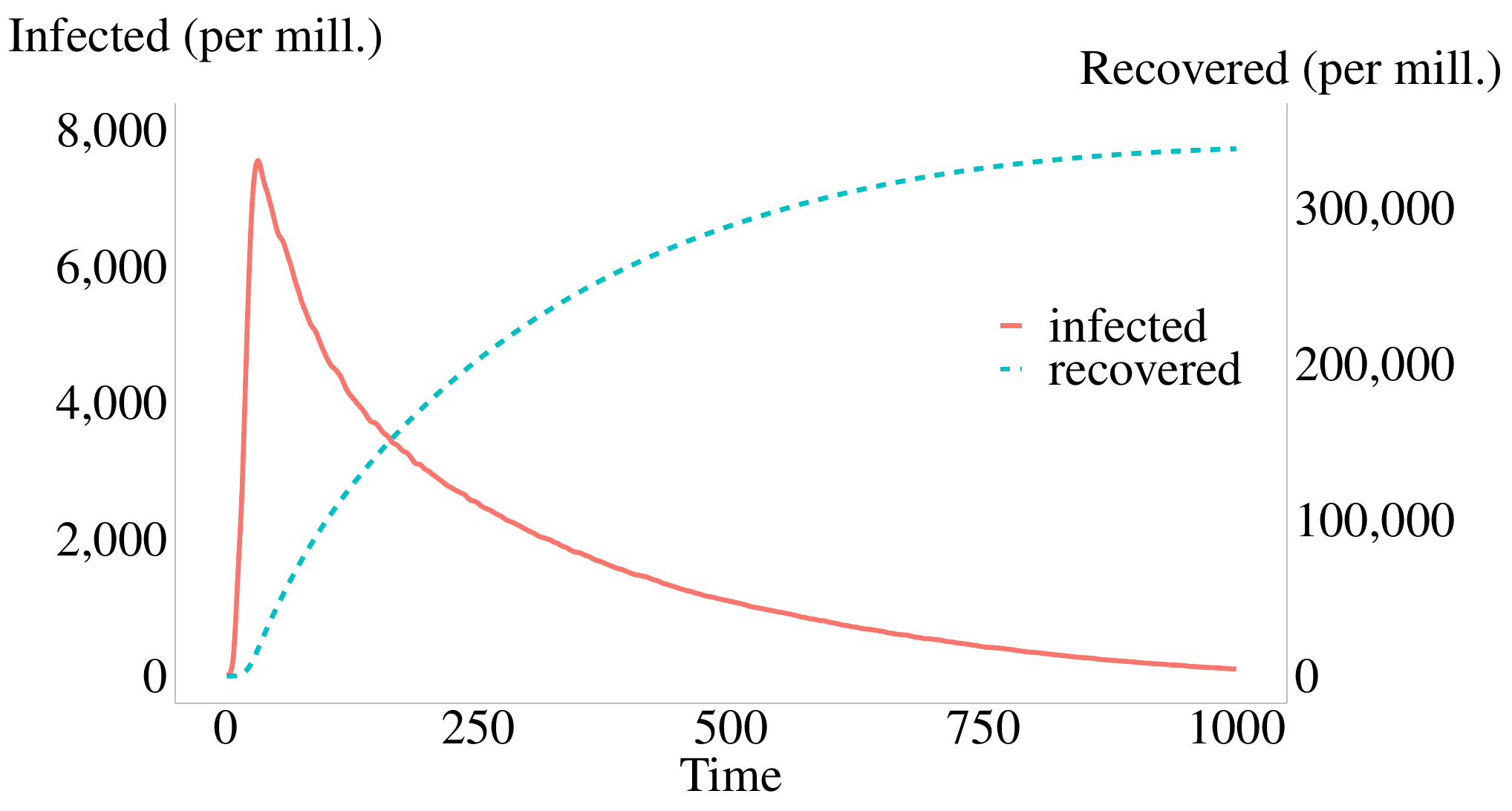}
	}
	\hfill
	\subfloat[36 jurisdictions proactively quarantine by estimating internal infections based on observation of
	other jurisdictions, 4 lax jurisdictions only quarantine once they reach 5 internal infections\label{subfig-4:adaptive}]{%
		\includegraphics[width=0.4\linewidth]{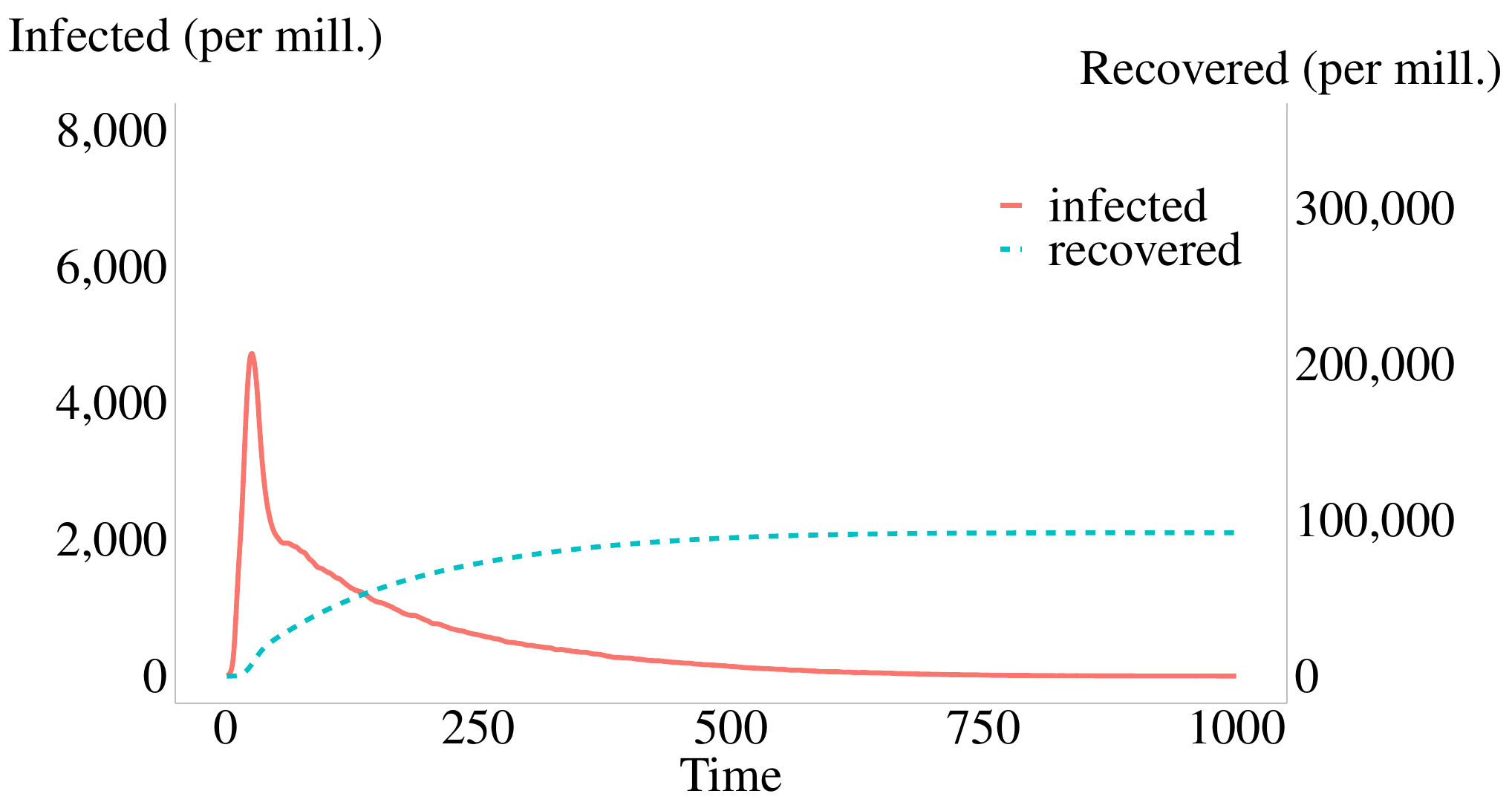}
	}
	\caption{ We picture daily infections and cumulative recoveries under four quarantine policies with
		40 jurisdictions. When a jurisdiction quarantines, it locks down the entire jurisdiction.
		In Panel \ref{subfig-1:adaptive}, all jurisdictions use a reactive policy.
		In Panel \ref{subfig-2:adaptive}, all jurisdictions use a proactive policy.
		In Panel \ref{subfig-3:adaptive}, we implement the same policies as Panel \ref{subfig-1:adaptive},
		but have four lax jurisdictions that use $x=5$ (0.14\% of the jurisdiction population) instead of $x=1$. Panel \ref{subfig-4:adaptive} has
		36 jurisdictions with proactive policies and four with lax policies.
		For each figure, we simulate 10000 times on the same network with random initial infections,
		and present the average number of infections and recovered people over time, scaled per million. }
	\label{fig:naive_vs_proactive}
\end{figure*}

Figure \ref{subfig-1:compare_policies} displays the dynamics of quarantines for each of the
policy configurations from Figure \ref{fig:naive_vs_proactive},
and Figure \ref{subfig-2:compare_policies} displays
the number of person-days of infection versus the number of person-days of quarantine.
Single jurisdiction policies (global quarantines and $(k,x)$ regional quarantines
(with $\varepsilon=0.05$ leakage)) do the best on both dimensions. Once multiple jurisdictions are introduced, proactive policies perform better than reactive ones for both infected and quarantined person-days. Even with lax jurisdictions, the proactive policy is better than the reactive policy  without lax jurisdictions. For both proactive and reactive policies, introducing lax jurisdictions increases infected person-days. However, this effect is not uniform: with proactive jurisdictions, lax jurisdictions cause a larger increase in infections (both absolutely and proportionally). With respect to quarantined person-days, lax jurisdictions have different effects depending on other jurisdictions' policies. With proactive jurisdictions, quarantined person-days increase, while with reactive jurisdictions, they slightly decrease. With reactive policies, infections spread so rapidly from lax jurisdictions that by coincidence, large numbers of jurisdictions quarantine at once -- and eventually an almost global quarantine occurs, halting the disease more quickly than in the scenarios without lax jurisdictions, but with higher infections. This relative ordering of quarantined-person days for reactive policies depends on parameters, as demonstrated in the SI: increasing $\alpha$ causes the number of quarantined person-days to increase once lax jurisdictions are added to reactive policies, rather than decrease.

\begin{figure*}[h]
	\centering
	\textbf{Figure \ref{fig:compare_policies}: The Impact and Costs of Quarantine Policies with and without Lax Jurisdictions}\par\medskip
	\subfloat[ Dynamics of quarantines in each of the policy configurations \label{subfig-1:compare_policies}]{%
		\includegraphics[width=0.48\linewidth]{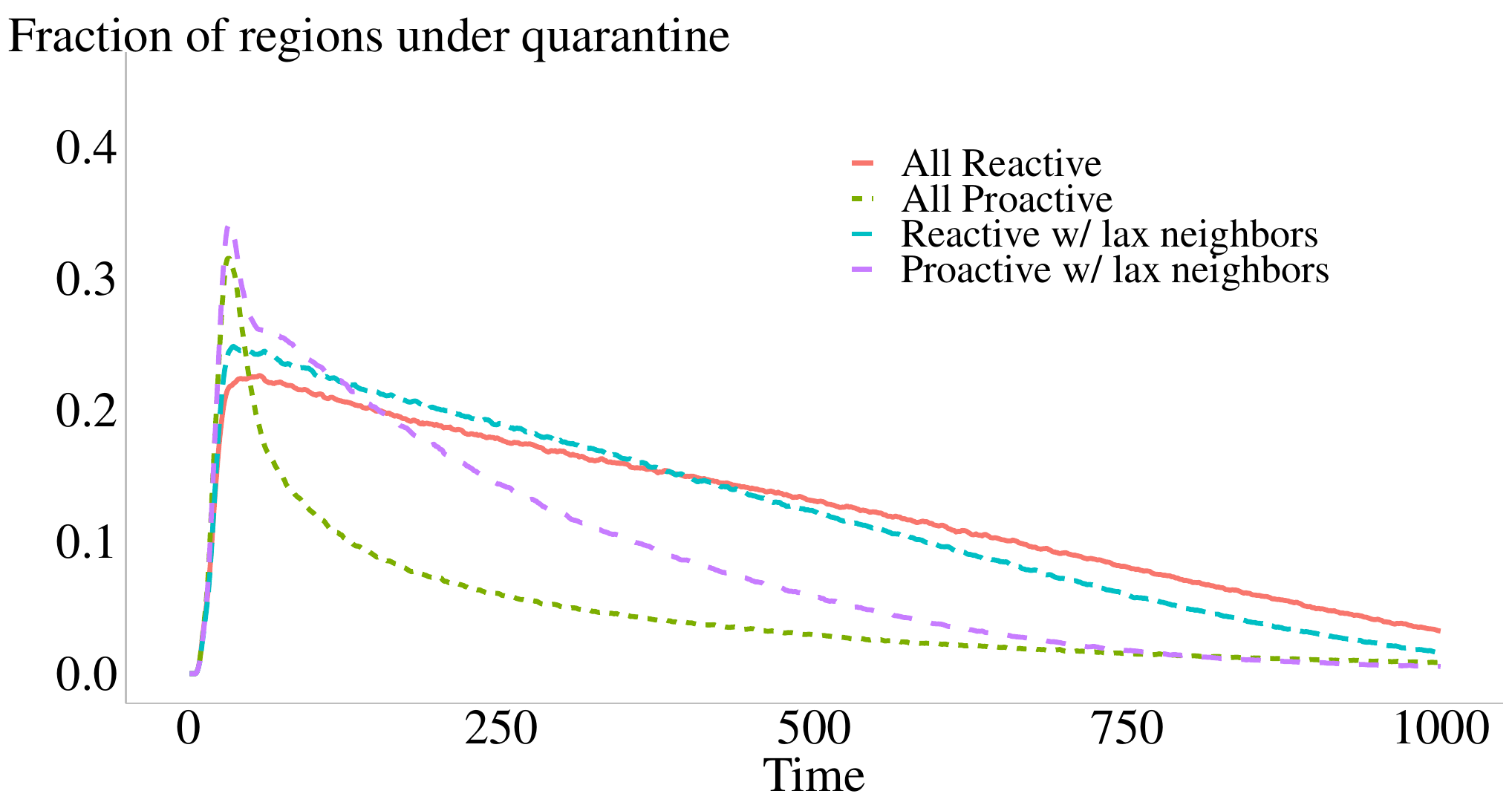}
	}
	\hfill
	\subfloat[Person-day infections vs. person-day quarantines (per million) \label{subfig-2:compare_policies}]{%
		\includegraphics[width=0.48\linewidth]{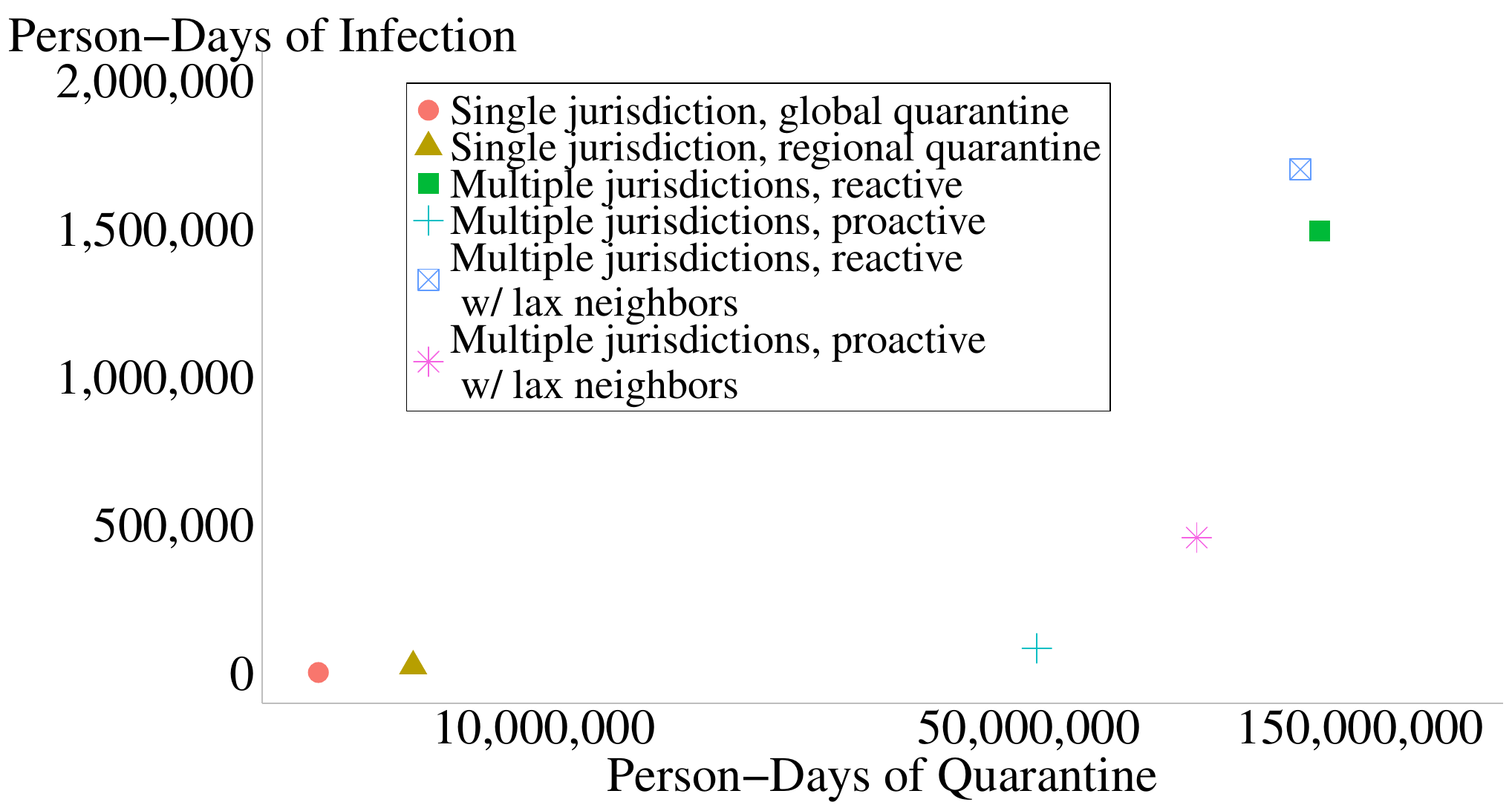}
	}\\
	\caption{Figure \ref{subfig-1:compare_policies} displays the dynamics of quarantines for each of the
		policy configurations. Figure \ref{subfig-2:compare_policies} plots the number of person-day
		infections (per million) against the number of person-day quarantines (per million) for six
		key policy scenarios.
		The global policy does the best on both dimensions, and the second best is the single-jurisdiction reactive strategy
		(which does worse than the global because of leakage).
		With 40 jurisdictions, both proactive policies outperform the internal, reactive policies. By far the worst, on both dimensions, is the internal, reactive policy with some lax jurisdictions.
		These results come from the same simulations that produce figures  \ref{fig:theorem_conditions} and \ref{fig:naive_vs_proactive}.}
	\label{fig:compare_policies}
\end{figure*}

\section*{Discussion}

We have shown that regional quarantine policies are likely to fail to halt the spread of a virus in most empirical settings, unless there is extremely rapid and efficient detection of the disease, and governments can halt all contact within the quarantined region. This failure is due to the failure of what we call ``growth-balance'', which ensures that there are no infection paths leading from infected individuals to others outside the quarantined region that are likely to be undetected.  Multiple
governments using independent policies are even less effective, as leakage occurs across their borders.  We have also shown that if governments are more attentive to their neighbors, there can be substantial improvements to their infection rates. However, if some jurisdictions are lax, this imposes a significant cost on everyone else. 

Jurisdictional policies tend to be aimed at the welfare of their internal populations,
yet the external effects are large.
Our results underscore the importance of timely information sharing and coordination in
both the design and execution of policies across jurisdictional boundaries \cite{elliott2019network}.
The results also underscore the global importance of aiding poor jurisdictions.
Indeed, there is mounting evidence that a lack of coordination across boundaries has been damaging in the case of COVID-19 \cite{holtz2020interdependence}.

The use of masks (decreasing $p$), social distancing (decreasing $d$),
and increasing testing (increasing $\alpha$), and vaccinations (decreasing $p$), all help attenuate
contagion, but unless they maintain the reproduction number below one,
the problems identified here remain.  Even tiny fractions of
interactions across borders are enough to lead to spreading in
large populations.  With modern inter- and intra-national trade and travel being a
sizable portion of all economies, such interaction is difficult to
avoid. 
Nonetheless, our analysis also offers insights into managing infections at
smaller scales; e.g., within schools, sports, and businesses.  By creating a
network of interactions that is highly modular, keeping cross-modular interactions to a minimum and making sure that they are highly traceable, together with aggressive testing (especially of
cross-module actors), one can eliminate leakage and effectively bound the set of interactions. This would divide the network into small components of diameter less than $k$, so that growth balance is satisfied by default.

We note that the effect of lax jurisdictions can be mitigated if other jurisdictions eliminate contact with that jurisdiction; e.g., with travel restrictions.
Then, for the non-lax jurisdictions, the situation returns to one without any lax jurisdictions, which is better as we have shown.
This will only work if all non-lax jurisdictions participate in a travel ban, as otherwise infection will continue to resurge in jurisdictions that continue contact with the lax jurisdictions, which then pass infection along to others. 

Our results also suggest caution in using statistical models to
identify regions to quarantine.  Although contagion models are
helpful for informing policy about the magnitude of an epidemic and
broad dynamics, the models can give false comfort in our
ability to engage in highly targeted policies, whose results can be
influenced by small deviations from idealized assumptions (e.g., leakage). Our growth-balance condition also points out that not all parts of a
network are equal in their potential for undetected transmission. Growth balance offers insight into when containment will frequently fail. While full containment is sometimes, but not always, the full policy goal, it helps us understand what features of the network aid or hinder containment efforts even under ideal conditions. Policy makers must be conscious that limited paths of interactions that can introduce the disease into a larger population, precisely because detection is very difficult along such paths, and therefore are important to monitor. In
places where the reproduction number is lower, the probability of observing outbreaks is also lower,
enabling a leakage of undetected infections.

\bibliographystyle{naturemag}
\bibliography{lockdown}

\clearpage
\appendix


\setcounter{table}{0}
\renewcommand{\thetable}{C.\arabic{table}}
\setcounter{figure}{0}
\renewcommand{\thefigure}{C.\arabic{figure}}

\begin{center}
	{\Large {\bf Supplementary Material} \\
Interacting Regional Policies in Containing a Disease \\
by Chandrasekhar, Goldsmith-Pinkham, Jackson, Thau}
\end{center}

\section{Theorem Details and Proofs}

\subsubsection*{People and  Interactions}

There are $n > 1$  nodes (individuals) in an unweighted, and possibly directed, network.

We study the course of a disease through the network. Time is discrete, with periods indexed by $t \in \mathbb{N}$.  An initial infected node, indexed by $i_0 \in V$, is the only node
infected at time 0.  We call this node the \emph{seed}.

We track the network via neighborhoods that expand outwards via (directed) paths from $i_0$.
Let $N_k$ be all the nodes who are at (directed) distance $k$ from node $i_0$.
Let $n_k$ denote the cardinality of $N_k$.

For any node in  $j\in N_{k'}$, for  $k'< k$,  let $n^j$ be the number of its direct descendants and  $n_k^j$ be the number of its (possibly indirect) descendants in  $N_k$ that are reached by never passing beyond distance
$k$ from $i_0$.

Unweighted network models are admitted here.  Additionally, the results below extend to any weighted model in which weights are bounded above and below {(e.g., probabilities of interaction)}.  Note also, that the network can be directed or undirected.

The infection process proceeds as follows. In every time period $t\in \{1,2,\ldots\}$, an infected node $i$ transmits the disease to each of $i$'s neighbors  independently with probability $p$. A newly infected node is infectious for $\theta \geq 1$ periods after which the node recovers and is never again infectious. The model can easily be extended to accommodate renewed susceptibility.

There may be a \emph{delay} in the ability to detect the disease. The number of periods of delay is given by $\tau$ with $0 \leq \tau \leq \theta$. Delay is a general term that can capture many things. For example, it can correspond to (a) asymptomatic infectiousness, (b) a delay in accessing health care given the onset of an infectious period, (c) any delay in the administration of testing, and so on.

In the first period of an infected node's infectious period -- after delay ($\tau$) -- there is a
probability $\alpha$ that the policymaker detects it as being infected.
So, potential detection happens exactly once during the first period in which the node can be detected.
Detection is independently and identically distributed. Our results are easily extended to have a random period for detection after the delay.

Finally, the policymaker may face some error in their knowledge of the network. This can come from their limited enforcement capacity, random noise in data collected to estimate interaction networks, or from network model misspecification. If there is error, we will track a share $\epsilon$ of nodes that are within a $k$-neighborhood of the seed but are estimated by the policymaker to be outside the $k$-neighborhood.

\subsubsection*{Regional Quarantine Policy}

Let a \emph{regional policy of distance $k$ and threshold $x$} be such that  once there are at least $x$ infections (other than the seed) detected within distance $k$ from the
initial seed, then all nodes within distance $k+1$ of $i_0$ are quarantined for at least $\theta$ periods.  A quarantine implies all connections between nodes are severed to avoid any further transmission, the infection lasts its duration $\theta$ and dies out.

Implicit in this definition is that a quarantine is not instantaneous, but that infected people could have infected their neighbors before being shut down, which is why the nodes at distance $k+1$ are quarantined.
All the results below extend if we assume that it is instantaneous, but with quarantines moved back one step and path lengths in definitions correspondingly adjusted.

We have assumed the policymaker knows the ``seed,'' for simplicity - and which may take some time in reality.
This provides an advantage to the policymaker, but we see substantial containment failures despite this advantage.

\subsubsection*{Growth Balance}

In order to conduct asymptotic analysis, a useful device to study the probabilities of events in question in large networks, we study a sequence of networks $G(n)$ with $n\rightarrow \infty$ and an associated sequence of  parameters $(\alpha,p,\tau,\theta,{k}) = (\alpha(n),p(n),\tau(n),\theta(n),k(n))$. 

Consider a network and a distance $k$ from the initially infected node $i_0$.
A {\sl path of potential infection to} ${k+2}$ is
a sequence of nodes $i_0,i_1,\ldots i_\ell$ with $i_\ell\in N_{k+1}$, $i_{j+1}$ being a direct descendant of $i_j$ for each $j\in \{0,\ldots, \ell-1\}$, and for which $i_\ell$ has a descendant in $N_{k+2}$.

Consider a sequence of networks and $k(n)$s.
We say that there are {\sl bounded paths of potential infection  } from $i_0(n)$ to $k(n)+2$ if there exists some finite $M$
and for each $n$ there is a path of potential infection to ${k(n)+2}$,
$i_0(n),i_1,\ldots i_\ell$ of length less than $M$,  with  $n^{i_j}<M$ for every $j\in \{0,\ldots, \ell-2\}$.

We say that a sequence of networks is {\sl growth-balanced} relative to some $k(n)$ (and sequence of $i_0(n)$)
if there are no bounded paths of potential infection to $k(n)+2$.   This is equivalent to stating that there exists a sequence $m(n)\rightarrow \infty$ such that each path of potential infection from $i_0(n)$ to $k(n)+2$ is either of length at least $m(n)$ or has some node with degree at least $m(n)$. 

If $k(n)$ grows without bound, then the condition is satisfied trivially, so the bounded case is the one of interest; and also the one of practical interest given the small diameter of real-world networks.

Also note that the condition is stated with respect to a sequence of seed nodes.  The results extend directly if one wants things to hold with respect to sets of seeds by requiring that the conditions hold for sequences of sets of seeds.

Growth balance is essentially a condition that requires a minimum bound of expansion along all potential paths of infection to escape a regional quarantine from some initial infection.
The intuition behind the condition is clear:  to ensure detection of an outbreak before it reaches a distance $k+1$ from the seed, many of the nodes within distance $k$ must be exposed to the disease by the time it reaches distance $k$.   What is ruled out is
a relatively short path that gets directly to that distance without having many nodes be exposed along that path.\footnote{This is very different from conditions that concern long paths within
short distances, such as \cite{ugander2013graph}, as ours is ruling out {\sl short} paths with low expansion.}

Figure S1 presents an illustration of a network that is not growth-balanced.

\subsection*{Results}

\subsubsection*{A Benchmark: No Delay in Detection; Perfect Information and Enforcement}

We begin with a benchmark case in which there is no delay in detection  ($\tau(n) = 0$) and the policymaker can completely enforce a quarantine at some distance $k(n)+1$.\footnote{Note that this requires knowledge
of the neighborhood structure around the seed node, but no other
knowledge of the network by a policy maker.}

We allow the size of the quarantine region $k$ to depend on $n$ in any way, as the theorem still applies.
We work with an arbitrary but fixed infection threshold $x$.  What is important
is that $x$ not grow too rapidly, as otherwise the likelihood of observing $x$ infections within proximity $k$ to the seed is extremely low.\footnote{
The theorem extends to allow $x=x(n)$ to grow with $n$, provided the growth is sufficiently slow, and then that growth-balance condition becomes more
complicated, as the $M = M(n)$ in that definition adjusts with the rate of growth of $x$.}

\begin{theorem}\label{thm:benchmark}
	Consider any sequence of networks and associated $k(n)<K(n)-1$ where $K(n)$ is the maximum $k(n)$ for which $n_{k}>0$;\footnote{Otherwise, it is actually a global policy.}$^,$
	such that each node in $N_{k(n)+1}$ has at least one descendent at distance $k(n)+2$, and let $x$ be any fixed positive integer.
	Let the sequence of associated diseases have $\alpha(n)$ and $p(n)$ bounded away from 0 and 1,\footnote{The cases of $p(n)$ or $\alpha(n)$ equal to 1 are degenerate.}
	no delay in detection, and any $\theta(n)\geq 1$.
	A regional quarantining policy of distance $k(n)$ and threshold $x$  halts all infections past distance $k(n)+1$ with a probability tending to 1
	\emph{if and only if} the sequence is growth-balanced with respect to $k(n)$.
\end{theorem}

Note that the growth-balance condition implies that the number of nodes within distance $k(n)$ from $i_0$ must grow without bound. 
Theorem \ref{thm:benchmark} thus implies that in order for a regional policy to work, the region size must grow without bound, and also must satisfy a particular balance condition.  (Rates at which this growth must occur as a function of $k$ and $n$, can be deduced from the relevant infection probabilities and network structure.)

\begin{proof}[Proof of Theorem \ref{thm:benchmark}]
	
	To prove the first part, note that if the infection never reaches distance $k(n)$ then the result holds directly since it can then not go beyond $k(n)+1$.
	We show  that if the sequence of networks is growth-balanced relative to $k(n)$, then conditional upon an infection reaching level $k(n)$ with the possibility of reaching $k(n)+2$ within two periods, the probability that it infects more than $x$ nodes within distance $k(n)$ before any nodes beyond $k(n)$
	tends to 1.
	Suppose that infection reaches some node at distance $k(n)$ that can reach a node in $N_{k+1}$.   Consider the corresponding sequence of paths of infected nodes
$i_0,i_1,\ldots i_\ell$ with $i_\ell\in N_{k+1}$, $i_{j+1}$ being a direct descendant of $i_j$ for each $j\in \{0,\ldots, \ell-1\}$, and note that by assumption $i_\ell$ has a descendant in $N_{k+2}$.
	By the growth-balance condition, for any $M$, there is a large enough $n$ for  which either the length of the path is longer than $M$ or else there is at least one $i_j$ with $j\leq \ell-2$ along the path that has more than $M$ descendants.
	In the latter case, the probability that $i_j$ has more than $x$ descendants who become infected and are detected is at least $1-F_{M,m}(x)$ where $F_{M,m}$ is the binomial distribution
	with $M$ draws each with probability $m$, where $p(n)\alpha(n)>m$ for some fixed $m$.  Given that $x$ and $m$ are fixed, this tends to probability 1 as $M$ grows.
	In the former case, the sequence exceeds length $M$, all of which are infected and so given that $\alpha(n)$ is bounded below, the probability that at least $x$ of them are detected goes to 1 as $M$ grows.
	In both cases, as $n$ grows, the minimal $M$ across such paths of potential infection to $k(n)+1$ grows without bound, and so
	the probability that there are at least $x$ infections that are detected by the time that $i_{\ell -1}$ is reached tends to 1 as $n$ grows.

	To prove the converse, suppose that the network is not growth-balanced.
	Consider a sequence of bounded paths of potential infection to $k(n)+2$, with associated
	sequences of nodes $i_0,i_1,\ldots i_\ell$ of length less than $M$ with $i_\ell\in N_{k+1}$, $i_{j+1}$ being a direct descendant of $i_j$ for each $j\in \{0,\ldots, \ell-1\}$, with  $n^j<M$ for every $j\in \{0,\ldots, \ell-2\}$, and for which $i_\ell$ has a descendant in $N_{k+2}$.
	The probability that each of the nodes $i_1,\ldots i_{\ell-2}$ becomes infected and no other nodes are infected within distance $k(n)-1$, and that all infected nodes are undetected is at least $(p(n)(1-\alpha(n)) (1-p(n))^M )^M$.
	This is fixed and so bounded away from 0.  This implies that probability that the infection gets to nodes at distance $k(n)$, and $i_{\ell-1}$ in particular, without any detections is bounded below.
Thus, there is a probability bounded below of reaching $i_\ell$ before any detections, and then by the time the quarantine is enacted,
there is at least a $p(n)$ times this probability that it
 escapes past $N_{k+1}$, which is thus also bounded away from 0.\end{proof}

We note that Theorem \ref{thm:benchmark} admits essentially all sequences of (unweighted) networks.
Thus, for every type of network, one can determine whether a regional policy of some $(k(n),x)$ will succeed or fail.
The only thing that one needs to check is growth-balance.
If it is satisfied, a regional policy works, and otherwise it will fail with nontrivial probability.

This has implications for some prominent random network models. 
Consider a randomly chosen sequence of seeds and networks from the associated networks:
  \begin{enumerate}
  \item For a sequence of stochastic block models in which all nodes have expected degree $d(n)> \log(n)$ so that the network is path connected (with Erdos-Renyi as a
    special case),\footnote{Consider a sequence of block models
      such that the ratio of expected out degree of a node in
      one neighborhood compared to another in some other block
      cannot grow without bound.} a regional policy with a
    bounded $k(n)$ has a probability going to 1 of halting the
    disease on the randomly realized network if and only if the
    seed node's expected out degree $d(n)>1$ is such that
    $d(n)^{k(n)}\rightarrow\infty$.
  \item For a regular expander graph with outdegree $d(n)> 1$,  a regional policy works if and only if  the expansion rate $d(n)^{k(n)}\rightarrow\infty$.
  \item For a regular lattice of degree $d(n)>1$,  a regional policy works if and only if $d(n)^{k(n)}\rightarrow \infty$.
  \item For a rewired lattice with $d(n)>1$ for all nodes and with a fraction links that are randomly rewired, a regional policy with a bounded $k(n)$ has a probability going to 1 of halting the disease on the randomly realized network if and only if $d(n)^{k(n)}\rightarrow \infty$.
  \item For a sequence of random networks with a scale-free degree distribution with average degree $d(n)>\log(n)$,
    a regional policy works (with probability 1) if and only if $k(n)\rightarrow \infty$.
  \end{enumerate}
  

Thus, whether a regional policy works in almost any network model requires that  either
the degree of almost all nodes grows without bound, or else the size of the quarantine grows without bound.
For a scale free distribution, there is always a nontrivial probability on small degrees, and hence in order for a regional policy to work, the size of the neighborhood must grow without bound.

In practice, even very sparse networks will have a large $d(n)^{k(n)}$ (e.g., if people have hundreds of contacts, $100^3$ is already a million and even with a very low $\alpha(n)$ many
infections will be detected within a few steps of the initial node).\footnote{This is still extremely sparse, as having 100 contacts out of millions or billions of potential other nodes is a small fraction.}
What the growth-balance condition rules out is that some nontrivial part of the network have
neighborhoods with many fewer contacts - so there cannot be people who have just a few contacts, since that will allow for a nontrivial probability of undetected escape (e.g., $2^3=8$ and so with only 8 infections, it is possible that none are detected and the disease escapes beyond 3 steps).
As many real-world network structures have substantial heterogeneity, with some people having very low numbers of interactions, such an escape becomes possible {even under idealized assumptions of no delay in detection and no leakage} \cite{hoffrh2002,penrose2003random,leskovec2008statistical,banerjeecdj2013,chandrasekharj2016}.

\subsubsection*{Delay in Detection}

The detection delay, $\tau(n)$, is distributed over the support $\{1,\ldots,\tau^{\max}(n)\}$. This includes degenerate distributions with $\tau^{\max}(n)$ being the maximal value of  the  support  with  positive mass. The policymaker may  or may  not know  $\tau^{\max}(n)$ and we study both cases. The latter is important as in practice we estimate delay periods so there is bound to be uncertainty.
When $\tau(n)$ is known, we can simply say $\tau(n) = \tau^{\max}(n)$.

Let a \emph{regional policy with trigger $k(n)$ and threshold $x$ and buffer $h(n)$} be such that once there are at least $x$ infections detected within distance $k(n)+h(n)$ from the
initial seed, then all nodes within distance $k(n)+h(n)+1$ of $i_0$ are quarantined/locked down  for at least $\theta(n)$ periods.

There are two differences between this definition of regional policy from the one considered before.  First, it is triggered by infections within distance $k(n)+h(n)$ (not within distance $k(n)$), and it also has a buffer in how far the quarantine extends beyond the $k(n)$-th neighborhood.

We extend the definition of growth balance to account for buffers.

Consider a network and a distance $k(n)$ from the initially infected node $i_0$ and an $h(n)\geq 1$.
A {\sl path of potential infection to} ${k(n)+h(n)+2}$ is
a sequence of nodes $i_0,i_1,\ldots i_\ell$ with $i_\ell \in N_{k(n)+h(n)+1}$, $i_{j+1}$ being a direct descendant of $i_j$ for each $j\in \{0,\ldots, \ell-1\}$.

Consider a sequence of networks, $n$, and associated $k(n), h(n)$.
We say that there are {\sl bounded paths of potential infection }to $k(n)+h(n)+2$ if there exists some finite $M$
and for each $n$ there is a path of potential infection to ${k(n)+h(n)+2}$,
$i_0,i_1,\ldots i_\ell$ of length less than $M$,  with  $n^j<M$ for every $j\in \{0,\ldots, \ell-h(n)-2\}$.
We say that a sequence of networks is {\sl growth-balanced} relative to some $k(n)$ and buffers $h(n)$
if there are no bounded paths of potential infection to $k(n)+h(n)+2$.

\begin{theorem} \label{thm:latent}
Consider a sequence of diseases have $\alpha(n)$ and $p(n)$ bounded away from 0 and 1,
	$\theta(n)\geq 1$, and have a detection delay distributed over some set $\{1,\ldots, \tau^{\max}(n)\}$ with $\tau^{\max}>1$
	(with probability on $\tau^{\max}(n)$ bounded away from 0).\footnote{A special case is in which $\tau^{\max}(n)$ is known.}
	Consider any sequence of networks
	and  $k(n)<K(n)-\tau^{\max}(n)-1$ where $K(n)$ is the maximum $k(n)$ for which $n_k>0$,
such that each node in $N_{k'}$ for $k(n)'>k(n)$ has at least one descendent at distance $k(n)'+1$,
	and let $x$ be any fixed positive integer.
	A regional policy with trigger $k(n)$, threshold $x$, and buffer $ \tau^{\max}(n)$  halts all infections past distance $k(n)+\tau^{\max}(n) {+1}$ with a probability tending to 1
	\emph{if and only if} the sequence is growth-balanced with respect to $k(n)$.
\end{theorem}

The Proof of Theorem \ref{thm:latent} is a straightforward extension of the previous proof and so it is omitted.

This result shows several  things. First, if the detection delay is small relative to the diameter  of the graph, one can use a regional  quarantine policy -- adjusted for the detection delay -- along the lines of that from  Theorem \ref{thm:benchmark} and ensure no  further  spread. This  is true even if the period  is  stochastic as long as the upper bound is known to be small.

Second,  and in contrast, if the detection delay is large compared to the  diameter of the graph,
then a regional  policy  is insufficient.  By the time  infections are observed, it is too late to
quarantine a subset of the graph.  This condition will tend to bind in the case  of real world networks,
as they exhibit small world properties and have small diameters
\cite{amaral2000classes,chung2002average}. 
As a result, even short detection delays may correspond to rapidly moving wavefronts that spread undetected.

\subsubsection*{Leakage in the Quarantine}

Next we turn to the case of in which there is some leakage in the quarantine, which may come for a variety of reasons.
First, the policymaker may have measurement error in knowledge of the network structure and thus who should be quarantined.  Second, and distinctly, lockdowns are
imperfect, and some transmission still happens.
Third, the network may cross jurisdictional borders and some nodes within distance $k(n)$ of $i_0$ may be outside of the policymaker's jurisdiction.

To keep the analysis uncluttered, we assume no detection delay, but the arguments extend directly to the delay case with the appropriate buffer.

\begin{theorem}\label{thm:knowledge}
	Consider any sequence of networks.
	Let the sequence of associated diseases have $\alpha(n)$ and $p(n)$ bounded away from 0 and 1,
	and be such that $\theta(n)\geq 1$, with no detection delay.
	Consider any $k(n)<K(n)-1$ where $K$ is the maximum $k(n)$ for which $n_k>0$, suppose that each node in $N_{k(n)}$ has at least one
descendent at distance $k(n)+1$, and let $x$ be any positive integer. 	
	
	Suppose that a random share of $\varepsilon_n $ of nodes within distance $k(n)$ of $i_0$ are not included in a regional quarantine policy and are connected
	to nodes of distance greater than $k(n)+1$ -- because of a lack of jurisdiction,
	misclassification by a policymaker, or lack of complete control over people's behaviors.\footnote{The misclassification can be that if some node within distance
	$k(n)$ is controlled by the quarantine, but connects to nodes that are not included and were thought to be of greater distance, but then allow the 
	disease to escape beyond the quarantine.}
	Then:
	\begin{enumerate}
		\item If $\varepsilon_n = o( (\sum_{k(n)'\leq k(n)} n_{k'})^{-1})$ and the network is growth-balanced, then a regional policy of distance $k(n)$ and threshold $x$ halts all infections past distance $k(n)+1$ with a probability  tending to 1.
		
		\item 	If $\varepsilon_n \geq \min[1 /x,\eta]$ for all $n$  for some $\eta>0 $
		or the network is not growth-balanced, then a regional policy of distance $k(n)$ and threshold $x$ fails to halt all infections past the regional quarantine with a probability bounded away from 0.	
	\end{enumerate}
\end{theorem}

\begin{proof}[Proof of Theorem \ref{thm:knowledge}]
	Part 1 follows from the fact that if  $\varepsilon_n = o( (\sum_{k(n)'\leq k(n)} n_{k'})^{-1})$ then
	the probability of having all nodes in $N_k$ correctly identified as being in $N_k$ tends to 1, and then Theorem \ref{thm:benchmark} can be applied.
	
	For Part 2, suppose that some $x$ infections are detected.
	The probability that at least one of them is misclassified is at least $1-(1-\varepsilon_n)^x$.
	Given that $\varepsilon_n \geq \min[1 /x,\eta
	]$ for some $\eta 
	>0$, it follows that $(1-\varepsilon_n)^x$ is bounded away from 1.
	There is a probability bounded away from 0 that at least one of the infected nodes is misclassified, and not subject to the quarantine, and connected to a node
	outside of distance $k(n)+1$. \end{proof}

	The theorem implies that the effectiveness of a regional policy is sensitive to any small fixed $\varepsilon$ amount of leakage.


\section{Simulation Details}
To illustrate the processes described in the main text, we run several simulations. First, we construct a large network with many jurisdictions. We directly study the content of the theorems with several versions of $(k,x)$ quarantines with an SIR infection process on a network. We use the same process and network to show the issues with jurisdictional policies, studying reactive and proactive policies.

\subsection*{Network Model}
We model the network structure as follows.
\begin{enumerate}
	\item There are $L$ locations distributed uniformly at random on the unit sphere. Each location has a population of $m$ nodes with a total of $n = mL$ nodes in the network.
	
	\item The linking rates across locations are given as in a spatial model \cite{hoffrh2002,breza2019ARD}. The probability of nodes $i \in \ell$ and $j \in \ell'$ for locations $\ell \neq \ell'$ linking depends only on the locations of the two nodes and  declines in distance:
	
	$$q_{\ell,\ell'}=\exp(a+b\cdot\textrm{dist}(\ell,\ell'))$$
	where $\textrm{dist}(\ell,\ell')$ is the distance between the two locations on the sphere and $a,b<0$.
	
	Every interaction between every pair of nodes is drawn independently from the observed spatial distribution, with distances measured along the surface of the unit sphere.
	
	\item The linking patterns within a location are given as in a mixture of random geometric (RGG) \cite{penrose2003random} and Erdos-Renyi (ER) random graphs \cite{erdosr1959}. Specifically, as spheres are locally Euclidean, we model nodes in a location (e.g., in a city) as residing in a square in the tangent space to the location. The probability that two nodes within a location link declines in their distance in this square.
	
We set $d_{RGG}$ as the desired degree from the RGG. Nodes are uniformly distributed on the unit
square $[0,1]^2$, and links are formed between nodes within radius $r_\ell$  \cite{penrose2003random}. Let $d_\ell$ be the desired average degree for all nodes within location $\ell$ with $m_\ell$ as the population at location $\ell$, which we take as fixed for our exercises.  	We define  
	$$\pi=\frac{d_{\ell}-d_{RGG}}{m_\ell}$$
	which is the probability with which remaining links within location are drawn (i.i.d.). 
 To obtain the desired degree we set
	
	$$r_\ell=\sqrt{\frac{d_{RGG}}{m_\ell\pi}}.$$

	\item Next, we uniformly add links to create a small world effect, with identical and independently distributed probability $s=\frac{1}{cn}$, where $c$ is an arbitrary constant and $n$ is the total number of nodes in the network \cite{watts1998collective}.
	
	\item Finally, we designate a single location as a ``hub,'' to emulate the idea that certain metro areas may have more connections to \textit{all} other regions. To do so, we select a hub uniformly at random and add links independently and identically distributed with probability $h$ from the hub location to every other location.
	
	\item To avoid the possibility of multiple links between the same two nodes, we remove any duplicate links.
\end{enumerate}

We first take $L=40$ and $m=3500$ for all locations. We set $a=-4$ and $b=-15$. Next, we calibrate the network to data by setting $d_\ell=15.5$, and $d_{RGG}=13.5$ for all locations. Next, we set $c=2$. Finally, we set $h=2.85\times 10^{-6}$. This process results in a graph that very roughly emulates the connectivity of real world networks in the United States and India \cite{mccormick2010many,banerjee2018changes,beaman2018can,banerjee2020messages}.
This includes data from India during the COVID-19 lockdowns about interactions within six feet,
meaning that it is conservative \cite{banerjee2020messages}.

We fix this network to use in all versions of the simulations. The network we generate
is sparse,  clustered, and has small average distances, as demonstrated by information detailed in Table S1.

Finally, we recalculate the connection probability matrix between locations to reflect rates of connection across regions, denoted by $q$. The entry that denotes the probability of linking between locations $\ell$ and $\ell'$ is $q_{\ell, \ell'}$.

\subsection*{Disease Process}
We set parameters as follows: the duration of infection is $\theta=5$, detection delay (when incorporated) is $\tau=3$, and set quarantine thresholds $x$ depending on the simulation.

We set transmission probability $p$ as

$$p =1-\left(1-\frac{R_0}{\bar{d}}\right)^{\frac{1}{\theta}}$$

where $\bar{d}$ is the mean degree.
 We take $R_0 = 3.5$, based on estimates of COVID-19 \cite{hao2020reconstruction}.

Following estimates from the literature (5-15\%), we set $\alpha=0.1$
   \cite{hortaccsu2020estimating,li2020substantial}.
In the simulations, each node is either detected or not during the first period in which it can be detected. Nodes that are detected are classified as such until recovery. Nodes that are undetected remain undetected (and so the $\alpha$ probability of detection is done once in the $\tau+1$st period, and only in that period).

As outlined in the main text, we begin by using $\theta=5$ and $\tau=3$ \cite{lauer2020incubation,hortaccsu2020estimating,li2020substantial, covidlength}.

\subsection*{Simulation Progression}

Each time period in the simulation progresses in four parts, which happen sequentially. The simulations run as follows:

\begin{enumerate}
    \item The policy makers see the newly detected infections from the previous period, and update their estimates of current infections (in all jurisdictions if proactive), and then determine whether a quarantine is necessary in their own jurisdiction in the next period (if there is not one already in place).  This quarantine decision is done based on estimated infections for proactive jurisdictions, and internally observed infections for reactive jurisdictions.
    \item The disease progresses for a period. This includes new infections and recoveries.
    \item Infected nodes that have just finished their detection delay of $\tau$ periods are independently detected with probability $\alpha$.
    \item New quarantines are enacted based on decisions made in step one of the process in this time period. Quarantines that have taken place for $\theta$ periods end.
\end{enumerate}

A node that becomes infected in period $t$ with a detection delay of $\tau$ and
total disease length $\theta$,
is tested in period $t+\tau$, results are processed in $t+\tau+1$,
and they will be quarantined (if necessary) starting at the end of $t+\tau+1$ (under the fourth item above).
This means that they have $\tau+1$ time periods during which they can infect other nodes. For instance, if $\tau=0$ this allows a node that becomes infected (but that was not already under quarantine for other reasons) one
opportunity to infect others.
This process reflects that neither detection nor quarantining of
individuals (or jurisdictions) happens instantaneously.
In addition, we stipulate that the seed node, $i_0$ is not counted in the quarantining
testing and calculations.
This is meant to reflect that it may be unclear whether the disease is spreading or not.
Nodes that are detected are marked as such until recovery.

\subsection*{Containment Policies}

A random node $i_0$ is selected and the epidemic begins there. We study the epidemic curve, the number total node-days of infection, and the number of node-days of quarantine for a variety of containment strategies. In all cases, the policy maker does not detect $i_0$, to emulate the difficulty of detecting an infection seed in real time. 

\subsubsection*{$(k,x)$ Policies}

We examine a number of scenarios using the $(k,x)$ policy model outlined in Theorems 1-3.

If a quarantine fails, and there are infections outside of the quarantine radius, the policy maker deals with each escaped infection individually. The policy maker treats each detected case outside of the initial quarantine as a new seed, and immediately quarantines all nodes with the same radius as the initial quarantine.

We begin pick our threshold for triggering the initial quarantine by using a simple objective function. We minimize a linear combination of the number of infected person periods and quarantined person periods. For all linear combinations where some weight is given to both terms, the optimal threshold is $x=1$. The logic is as follows: if the initial quarantine is successful, the number of quarantined person periods will be fixed and also the minimum number of quarantined person periods. Therefore, the problem reduces to minimizing the number of infections, which is done by setting $x=1$.

We study three versions of a $(k,x)$ policy. First, we simulate $(k,x)=(3,1)$ with no detection delay and no buffer. Then, we
incorporate a detection delay of $\tau=3$, still using a policy of $(k,x)=(3,1)$ and still with no buffer. (We do not include a buffer since the resulting quarantine on our network would encompass 99.98\% of nodes on average, since almost all nodes are within distance 6 of each other.) Lastly, we study a $(3,1)$ policy with enforcement failures and no buffer. In this case, a fraction $\epsilon=0.05$ of nodes do not ever quarantine. 

While the policy maker is unable to detect the infection seed in real time, once the policy maker decides to quarantine, we give them the advantage of perfect information with knowing the location of $i_0$. 

\subsection*{$(k,x)$ Policies with an Unknown Seed}

We also simulate the case where the policy maker is unable to trace back to find the initial seed $i_0$ to use as the center of the quarantine region. In this case, once at least $x$ cases are detected, the policy maker calculates the pairwise distance between the set of all detected nodes. The most central node is defined as the one with the minimum average distance to the other detected nodes. The policy maker then quarantines all nodes within distance $k+1$ of the most central node. If there are multiple nodes with the same average distance, the policy maker picks one at random. If the initial quarantine fails, the policy maker proceeds the same way as when they do know $i_0$, instituting quarantines of radius $k+1$ around detected nodes.

Again, we examine three cases: the first with $(k,x)=(3,1)$ with no detection delay, the second introducing a delay, and the third including enforcement failures. In the third case, a fraction $\epsilon = 0.05$ nodes never quarantine just as with the standard $(k,x)$ policies. Again, we do not include a buffer in any of the simulations as it would result in nearly global quarantines.

\subsubsection*{A Global Quarantine Policy}

In a global quarantine policy, every node is quarantined for $\theta$ periods as soon as at least $x=1$ infections are detected globally. We study this in the case with a detection delay, to compare it to the $(k,x)$, reactive, and proactive policies.

\subsubsection*{Reactive and Proactive Quarantine Policies}

For both the reactive and proactive policies, we take each location on the graph to be a separate jurisdiction.

\paragraph*{Reactive Quarantine Policies.}
Reactive jurisdictions respond only to detected infections within their own borders. We set $x=1$ for all jurisdictions,
the most conservative possible threshold, unless otherwise specified.

\paragraph*{Proactive Quarantine Policies.}
Proactive jurisdictions quarantine based on estimated
infection rates within their own borders, but these estimates account for the history of infections observed in all jurisdictions and knowledge of the network connection rates. In each period, each jurisdiction $\ell$ observes the number of actual detected infections at time $t$, $z_{\ell, t}$, and then calculates their estimated infections $w_\ell$ as follows:
$$w_{\ell,t}=\textrm{max}\{w_{\ell,t-1}+y_{\ell,t}-r_{\ell,t},z_{\ell,t}\},$$
where $y_{\ell,t}$ denotes the number of expected new infections in region $\ell$ at time $t$, given the history of infections observed in all jurisdictions and knowledge of the network connection rates, and $r_{\ell,t}$ denotes the number of expected recoveries in $\ell$ at $t$. The max updates the infection rate upwards if the estimated infection rate is lower than the actual observation.  This is not fully sophisticated, as the adjustment
could also backwardly update previous infection rates in light of the new information, but this would require introducing a probability space and more machinery that might improve the proactive policy's accuracy, but would not qualitatively change the results.

Each jurisdiction calculates $y_{\ell,t}$ as:
$$y_{\ell, t}=p\sum_{\ell' \textrm{ s.t. } \ell' \textrm{ not quarantined at t-1 }} m_{\ell'}q_{\ell,\ell'} w_{\ell', t-1} $$
The summation includes the term for spread from $\ell$ to still within $\ell$.
If $\ell$ is quarantined at time $t$, then $y_{\ell,t}=0$. Expected recovery at each
period $r_{\ell,t}$ is calculated as:
$$r_{\ell,t}=w_{\ell,t-\theta}-w_{\ell,t-\theta-1}+r_{\ell,t-\theta}.$$
Finally, we set $w_{\ell, t}<0.01$ to be zero, to avoid implementation issues
with floating point calculations. Setting a lower value to truncate at would
improve the performance of the proactive jurisdiction policies, as they would
be more sensitive to detected cases in other jurisdictions. We set $w_{\ell, 1}=0$, for all jurisdictions. Thus, the $w_{\ell, t}$ values remain at zero until at least one infection is detected somewhere.

\paragraph*{Uniform and Lax Policies}
We run two simulation variants for both the proactive and reactive policies: one in which all states are as conservative as possible, setting $x=1$ and a second in which
four regions set a higher threshold of $x=5$. In the proactive case, the lax jurisdictions follow a reactive policy in addition to using the higher threshold value. 

We choose $x=5$ to simulate lax thresholds.
In the United States, New York state issued a stay at home order when 0.07\% of the
state population was infected, which scaled to our populations of 3500 that is equivalent to a
threshold of 2.73 \cite{nyorder, covidtracking}.
When scaled to match our population of 3500, Florida began re-opening with a
threshold of 6.15, and some countries never locked down \cite{floridaorder, covidtracking, swedeninfo}.
The quarantines in our stylized model are more aggressive, as they cut contact completely.

\subsection*{Results and Sensitivity Analysis}
We run 10000 simulations with the parameters detailed in the main text: using $\theta=5$, $\tau=3$, $\alpha=0.1$ and $R_0=3.5$. Each simulation begins with a singular infection, selected uniformly at random. In the simulations where there are lax jurisdictions, four of the forty are selected to be lax uniformly at random. 
For all the additional sets of parameters reported below, we run 2500 simulations. 

We include the results of the simulations detailed in the main text in the tables below.
In addition, we run simulations with several sets of varied parameters: first, we take $\alpha=0.05$ and $\alpha =0.2$;
second we take $\theta=8$ and $\tau=5$; finally, we set $R_0$ equal to 2, 5, and 15 while holding all other parameters fixed.
Within the United States, estimates for the detection rate range from 5\% to 15\%,
and in countries with less developed testing infrastructure, the detection rate is
undoubtedly lower \cite{hortaccsu2020estimating}.
Because disease parameters are estimated, we use a different estimate of the disease lifespan of COVID-19 \cite{covidlength}. Full results are shown in Tables S2-S6, and in Figure S2. 

There are two key trends among the single regime policies. While the results from single jurisdiction policies in terms of infection and quarantined person-days are similar, regardless of whether or not the seed is known, knowing the seed node improves the effectiveness of the initial quarantine. This result is consistent with the theory. The similar results in terms of infections and quarantine person days is a result from the overall high effectiveness of the policy maker's response if the initial quarantine fails. Because the policy maker treats every escaped, detected infection as a new seed, no matter how it treated the initial quarantine, the overall results for infection and quarantine person-days are similar.

Second, as shown through the visuals of Figure S2, the effectiveness of the single policymaker policies varies depending on the disease parameters. For larger values of $R_0$, as demonstrated by the cases where $R_0 = 5$ and $R_0=15$, the single policymaker regional quarantine policies perform worse than the proactive, multiple jurisdiction simulations. With high values of $R_0$, the single jurisdiction policies perform better in terms of infected person-days, but have significantly more quarantine person days. This is because with a high $R_0$, precise targeting becomes much more difficult leading to many rounds of ineffective quarantines. In essence, the single jurisdiction is trying its best to halt the spread, but with a regional quarantine fails to get it under control. The multiple jurisdiction setting moves far more slowly, and so does not have the same number of quarantine days (but very similar infections). This is less a product of successful policy, and more a reflection that with $R_{0} = 15$, the only real effective policy is complete global quarantine.

There are several notable points about the reactive and proactive policies. First, the relationship between the reactive and proactive policies is robust to different sets of simulation parameters. In all cases, there is a significant gap between the proactive and reactive jurisdictions, along both the number of quarantine and infection person-days. Second, proactive policies are strictly better in terms of infections, regardless of the disease and administration parameters. Third, the impact of lax jurisdictions on quarantined person-days with reactive jurisdictions depends on the set of parameters. When $\alpha=0.2$, we see that adding lax jurisdictions to reactive policies increases the number of quarantine person-days. However, in all other cases, there are outcomes similar to those described in the main text: because of the high connection rate between jurisdictions, lax jurisdictions serve as super spreaders that cause coincidental large scale shut downs. Finally, lax jurisdictions uniformly increase the number of quarantined person-days for proactive jurisdictions.

\medskip

\section{Supplementary Tables}

\begin{table}[H]
\centering
\caption*{Table S1: Graph Statistics}
\begin{tabular}{rrr}
  \hline
 Property & Value \\
  \hline
Average Degree & 20.49 \\
Average Local Clustering Coefficient & 0.208 \\
Diameter & 9 \\
Average Path Length & 5.33 \\
\bottomrule
\end{tabular}

\caption*{Graph statistics for the graph used in all simulations. Similar to real world networks, it is sparse, clustered and has short average distances between nodes.}
\end{table}

\begin{table}[H]
\centering
\caption*{Table S2: Regional Policy (Known Seed) Simulation Results}
\begin{tabular}{ccccccccc}
    \toprule
     {$R_0$ } & {$\theta$} & {$\tau$} & {$\alpha$}  & {$\epsilon$} & {\shortstack{Percent\\Infected}} &{\shortstack{Infection\\Person Days}} & {\shortstack{Quarantined\\Person Days}} & {\shortstack{Escape\\Rate}}  \\ \midrule
    3.5 & 5  & 0 & 0.1  & 0 & 0.0276 & 1384.05 & 803955.61 & 0.0953 \\
   3.5 & 5  & 3  & 0.1  & 0 & 0.226 & 11282.19 & 2301413.60  & 0.458  \\
   3.5 & 5  & 3  & 0.1   & 0.05& 0.514 &25688.08 & 6478054.64   & 0.551     \\ \midrule
    3.5 &5  & 0   & 0.05    & 0 & 0.0684 & 3421.10 & 11231131.73 & 0.225 \\
   3.5 & 5  & 3  & 0.05  & 0 & 2.81 &  140667.17  &20297075.03   &  0.623  \\
   3.5 & 5  & 3  & 0.05  & 0.05 & 7.80 & 390155.83 & 66067046.93 & 0.706    \\ \midrule
    3.5 & 5 & 0 & 0.2 & 0 &0.0097&483.96&698551.61&0.022\\
     3.5 &5 & 3 & 0.2 & 0 &0.064&3196.86&1024409.97&0.260\\
     3.5 &5 & 3 & 0.2 & 0.05 &0.096&4794.23&2027933.31&0.352\\ \midrule
    3.5 &8  & 0  & 0.1   & 0 & 0.0277 & 2213.92 & 1243574.65  & 0.0904 \\
   3.5 & 8 & 5   & 0.1    & 0 &0.285 & 22834.58 & 4187189.53  & 0.506  \\
   3.5 & 8 & 5  & 0.1   & 0.05 & 0.559 & 44709.41 & 10653981.92   & 0.582 \\ \midrule
    2 & 5 & 0 & 0.1 & 0 & 0.0611&3057.36&737563.73&0.102\\
     2 &5 & 3 & 0.1 & 0 &0.149&7473.47&931141.96&0.226\\
     2 &5 & 3 & 0.1 & 0.05 &0.156&7784.11&1185778.89&0.252\\
    \midrule
     5 &5 & 0 & 0.1 & 0 &0.027&1349.74&800273&0.074 \\
   5 & 5 & 3 & 0.1 & 0 &1.607&80333.14&10826074.49&0.622\\
   5 & 5 & 3 & 0.1 & 0.05 &6.795&339788.66&63996541.26&0.746\\
    \midrule
   15 & 5 & 0 & 0.1 & 0 &0.0849&4245.51&712741.24&0.004\\
   15 & 5 & 3 & 0.1 & 0 &74.615&3730767.57&115033210.96&0.998\\
   15 & 5 & 3 & 0.1 & 0.05 &75.809&3790452.76&187358575.83&1.000\\
    
    \bottomrule
\end{tabular}
\caption*{Results for the parameters used in the main text are the average over 10000 simulations.
Results for the parameters only used in this section are the average over 2500 simulations.
For all simulations, we set $k=3$ and $x=1$. Infection person days and quarantined person days
are scaled to be per million individuals. The escape rate is defined as the frequency with
which the disease escapes the initial quarantine.}
\end{table}

\begin{table}[H]
\centering
\caption*{Table S3: Regional Policy (Unknown Seed) Simulation Results}
\begin{tabular}{ccccccccc}
    \toprule
     {$R_0$ } & {$\theta$} & {$\tau$} & {$\alpha$}  & {$\epsilon$} & {\shortstack{Percent\\Infected}} &{\shortstack{Infection\\Person Days}} & {\shortstack{Quarantined\\Person Days}} & {\shortstack{Escape\\Rate}}  \\ \midrule
    3.5 & 5  & 0 & 0.1  & 0 &0.0181&903.66&907418.73&0.191 \\
   3.5 & 5  & 3  & 0.1  & 0 &0.209&10425.09&2474174.31&0.524  \\
   3.5 & 5  & 3  & 0.1   & 0.05&0.503&25147.99&6733925.78&0.597   \\ \midrule
    3.5 &5  & 0   & 0.05    & 0 &0.052&2602.33&1434418.54&0.349 \\
   3.5 & 5  & 3  & 0.05  & 0 &2.689&134472.07&20800069.91&0.676  \\
   3.5 & 5  & 3  & 0.05  & 0.05 &8.038&401918.94&69666034.17&0.736\\ 
   \midrule
     3.5 &5 & 0 & 0.2 & 0 &0.0086&428.81&748090.71&0.067\\
   3.5 & 5 & 3 & 0.2 & 0 &0.061&3056.09&1141315.7&0.329\\
   3.5 & 5 & 3 & 0.2 & 0.05&0.097&4873.17&2221095.64&0.416\\\midrule
    3.5 &8  & 0  & 0.1   & 0 &0.007&562.29&708870.95&0.09 \\
   3.5 & 8 & 5   & 0.1    & 0 &0.228&18237.83&4193253.6&0.555  \\
   3.5 & 8 & 5  & 0.1   & 0.05 &0.561&44886.81&11767810.56&0.616\\ \midrule
    2 & 5 & 0 & 0.1 & 0 &0.0095&477.09&669783.94&0.117\\
     2 &5 & 3 & 0.1 & 0 &0.029&1427.07&830410.44&0.244\\
     2 &5 & 3 & 0.1 & 0.05 &0.034&1680.93&996141.17&0.263\\
    \midrule
     5 &5 & 0 & 0.1 & 0 &0.024&1202.11&952173.23&0.207\\
   5 & 5 & 3 & 0.1 & 0 &1.807&90371.77&12607529.19&0.718\\
   5 & 5 & 3 & 0.1 & 0.05&7.074&353675.16&66430323.4&0.789\\
    \midrule
   15 & 5 & 0 & 0.1 & 0 &0.1239&6195.7&1248499.49&0.266\\
   15 & 5 & 3 & 0.1 & 0 &73.99&3699513.54&123395786.67&0.998\\
   15 & 5 & 3 & 0.1 & 0.05&75.482&3774094.57&191925875.01&1.000\\
    
    \bottomrule
\end{tabular}
\caption*{Results for the parameters used in the main text are the average over 10000 simulations.
Results for the parameters only used in this section are the average over 2500 simulations.
For all simulations, we set $k=3$ and $x=1$. Infection person days and quarantined person days
are scaled to be per million individuals. The escape rate is defined as the frequency with
which the disease escapes the initial quarantine.}
\end{table}

\begin{table}[H]
    \centering
    \caption*{Table S4: Global Policy Simulation Results}
    \begin{tabular}{ccccccc}
    \toprule
        {$R_0$} & {$\theta$} & {$\tau$} & {$\alpha$} &{\shortstack{Percent\\Infected}} &{\shortstack{Infection\\Person Days}} & {\shortstack{Quarantined\\Person Days}} \\ \midrule
        3.5& 5 &  3 & 0.1 &0.0778&3890.51&4730000.00 \\
        3.5&5 &  3 & 0.05 &0.1518&7591.79&4702000.00 \\
        3.5 & 5 & 3 & 0.2  &0.044&2199.23&4716000.00 \\
        3.5& 8 &  5 & 0.1 &0.0847&6772.11&7545600.00 \\
        2& 5 &  3 & 0.1 &0.0207&1035.57&3708000.00\\
        5& 5 &  3 & 0.1 &0.1937&9685.63&4922000.00\\
        15& 5 &  3 & 0.1 &5.1639&258192.87&5000000.00\\
    \bottomrule
    \end{tabular}
    \caption*{Results for the parameters used in the main text are the average over 10000 simulations.
    Results for the parameters only used in this section are the average over 2500 simulations.
    Infection person days and quarantined person days are scaled to be per million individuals. There are fewer quarantined person days on average with $\alpha=0.05$, rather than $\alpha=0.1 $ as there is a greater chance of the disease going completely undetected before dying out. }
\end{table}

\begin{table}[H]
    \centering
    \caption*{Table S5: Reactive and Proactive Policy Simulation Results}
    \begin{tabular}{cccccccc}
    \toprule
         {Policy} &{$R_0$} &{$\theta$} & {$\tau$} & {$\alpha$} & \shortstack{Percent\\Infected} &{\shortstack{Infection\\Person Days}} & {\shortstack{Quarantined\\Person Days}} \\ \midrule
        Reactive &3.5& 5 &  3 & 0.1 &29.89&1494552.99&131303637.50\\
        Proactive &3.5& 5 &  3 & 0.1 &1.71&85526.78&51328755.00\\ 
        \midrule
        Reactive &3.5& 5 &  3 & 0.05 &44.91&2245276.69&132165200.00\\
        Proactive &3.5& 5 & 3 & 0.05 &5.45&272265.56&73938850.00 \\ 
        \midrule
        Reactive &3.5& 5 & 3 & 0.2 &10.63&531685.96&59927450.00\\
        Proactive &3.5& 5 & 3 & 0.2 &0.56&27845.97&33333750.00\\
        \midrule
        Reactive &3.5& 8 & 5 & 0.1 &27.66&2213177.60&194385520.00\\
        Proactive &3.5& 8 & 5 & 0.1 &2.50&200369.76&35172320.00\\
        \midrule
        Reactive &2& 5 & 3 & 0.1 &2.03&101275.14&13037500.00\\
        Proactive &2& 5 & 3 & 0.1 &0.20&10104.91&5300750.00\\
        \midrule
        Reactive &5& 5 & 3 & 0.1 &50.74&2537078.07&157829850.00\\
        Proactive &5& 5 & 3 & 0.1 &4.14&206837.40&38021550.00\\
        \midrule
        Reactive &15& 5 & 3 & 0.1 &81.66&4082954.97&16475250.00\\
        Proactive &15& 5 & 3 & 0.1 &71.10&3555090.86&12125500.00\\
    \bottomrule
    \end{tabular}
    \caption*{Results for the parameters used in the main text are the average over 10000 simulations.
    Results for the parameters only used in this section are the average over 2500 simulations.
    For all simulations, every jurisdiction sets $x=1$. Infection person days and quarantined person
    days are scaled to be per million individuals. }
\end{table}

\begin{table}[H]
    \centering
    \caption*{Table S6: Reactive and Proactive Policies with Lax Jurisdictions Simulation Results}
    \begin{tabular}{ccccccccc}
    \toprule
         {Policy} &{$R_0$}&{$\theta$} & {$\tau$} & {$\alpha$} &\shortstack{Percent\\Infected} &{\shortstack{Infection\\Person Days}} & {\shortstack{Quarantined\\Person Days}}  & \shortstack{Low Threshold\\Case Fraction}\\ \midrule
        Reactive & 3.5& 5 &  3 & 0.1 &34.06&1702933.48&123051062.50&0.853\\
        Proactive &3.5& 5 &  3 & 0.1 &9.19&459435.15&87241700.00&0.724\\ 
        \midrule
        Reactive &3.5& 5 &  3 & 0.05 &46.16&2308090.83&111742450.00&0.867\\
        Proactive &3.5& 5 & 3 & 0.05 &21.65&1082529.07&133660250.00&0.757\\ 
        \midrule
        Reactive &3.5& 5 & 3 & 0.2 &19.71&985572.76&101214350.00&0.846\\
        Proactive &3.5& 5 & 3 & 0.2 &1.93&96453.34&36203450.00&0.748 \\ 
        \midrule
        Reactive &3.5& 8 & 5 & 0.1 &32.78&2622535.7&192823440.00&0.852\\
        Proactive &3.5& 8 & 5 & 0.1 &11.99&958939.79&138729440.00&0.761
        \\\midrule
        Reactive &2& 5 & 3 & 0.1 &6.76&337947.53&34866850.00&0.857\\
        Proactive &2& 5 & 3 & 0.1 &1.62&81059.37&19606550.00&0.810\\
        \midrule
        Reactive &5& 5 & 3 & 0.1 &53.23&2661322.79&139446200.00&0.863\\
        Proactive &5& 5 & 3 & 0.1 &10.67&533748.29&69943200.00&0.717\\
        \midrule
        Reactive &15& 5 & 3 & 0.1 &83.63&4181626.60&15395600.00&0.885\\
        Proactive &15& 5 & 3 & 0.1 &73.66&3683186.66&10949990.00&0.871\\
    \bottomrule
    \end{tabular}
    \caption*{Results for the parameters used in the main text are the average over 10000 simulations.
    Results for the parameters only used in this section are the average over 2500 simulations.
    For all simulations, 36 jurisdictions set $x=1$ and the remained set $x=5$.
    In the proactive case, jurisdictions with $x=5$ follow  reactive policies.
    Infection person days and quarantined person days are scaled to be per million individuals.}
\end{table}

\bigskip

\section{Supplementary Figures}
\begin{figure}[H]
	     \textbf{Figure S1: Growth Balance}\par\medskip
	\centering
	\subfloat[Regional Policy Fails]{ \includegraphics[scale=0.15]{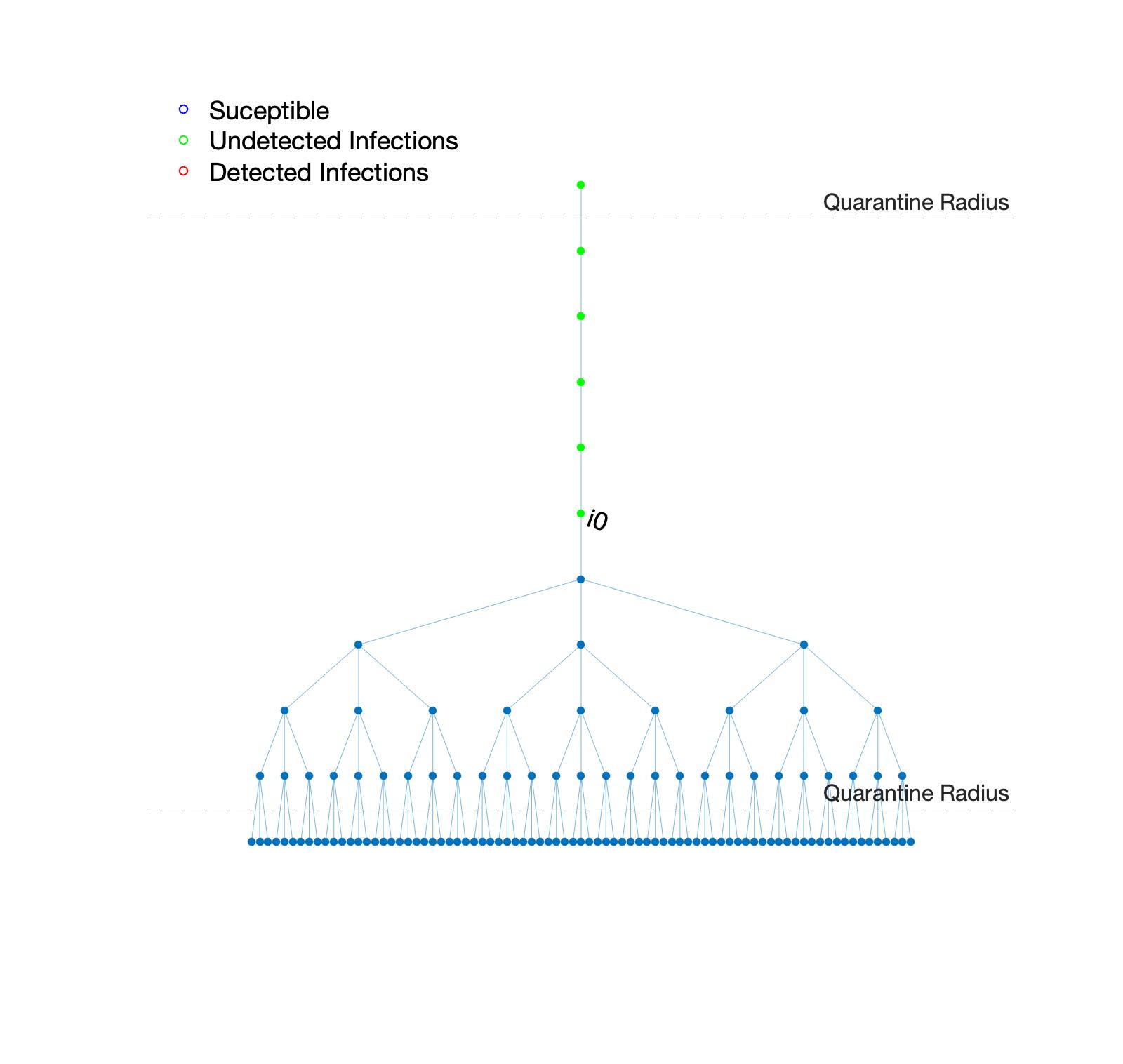}}
	\subfloat[Regional Policy Succeeds]{ \includegraphics[scale=0.15]{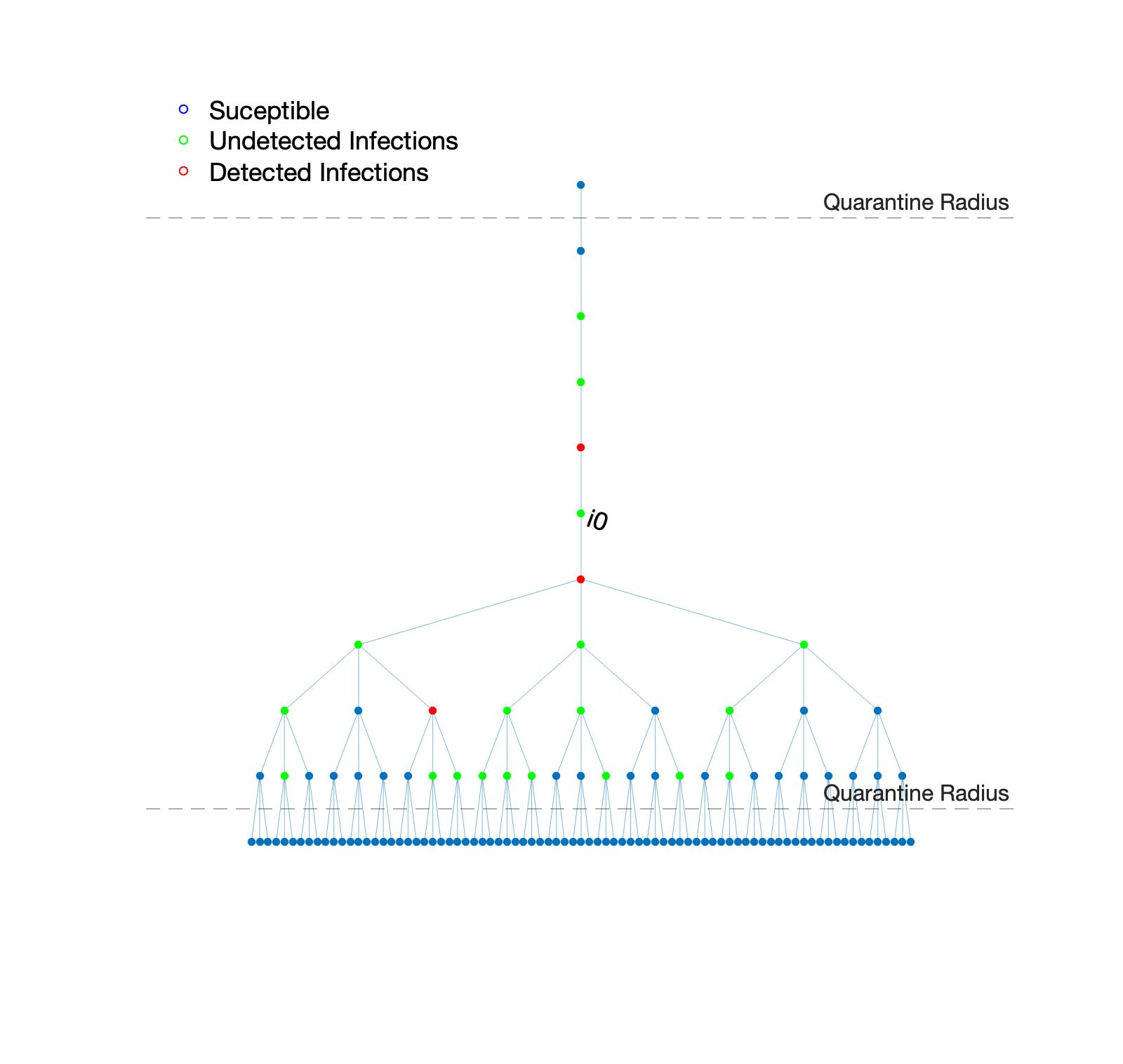}}
	\caption*{\textbf{Figure S1}: \small{Panel (a) demonstrates the possible failure of growth-balance.  The infection escapes up the line undetected beyond the quarantine radius.  If the infection happens to spread downwards, as in Panel (b), it is much more likely to be detected.  However, that only happens with some moderate probability in this network, and so growth balance fails.
	}}
	\label{fig:growthbalance}
\end{figure}

\setcounter{subfigure}{0}

\begin{figure}[H]
 \textbf{Figure S2: Impact and Costs of Quarantines with Different Simulation Parameters}\par\medskip

    \centering 
    \subfloat[]{\includegraphics[width=0.45\linewidth]{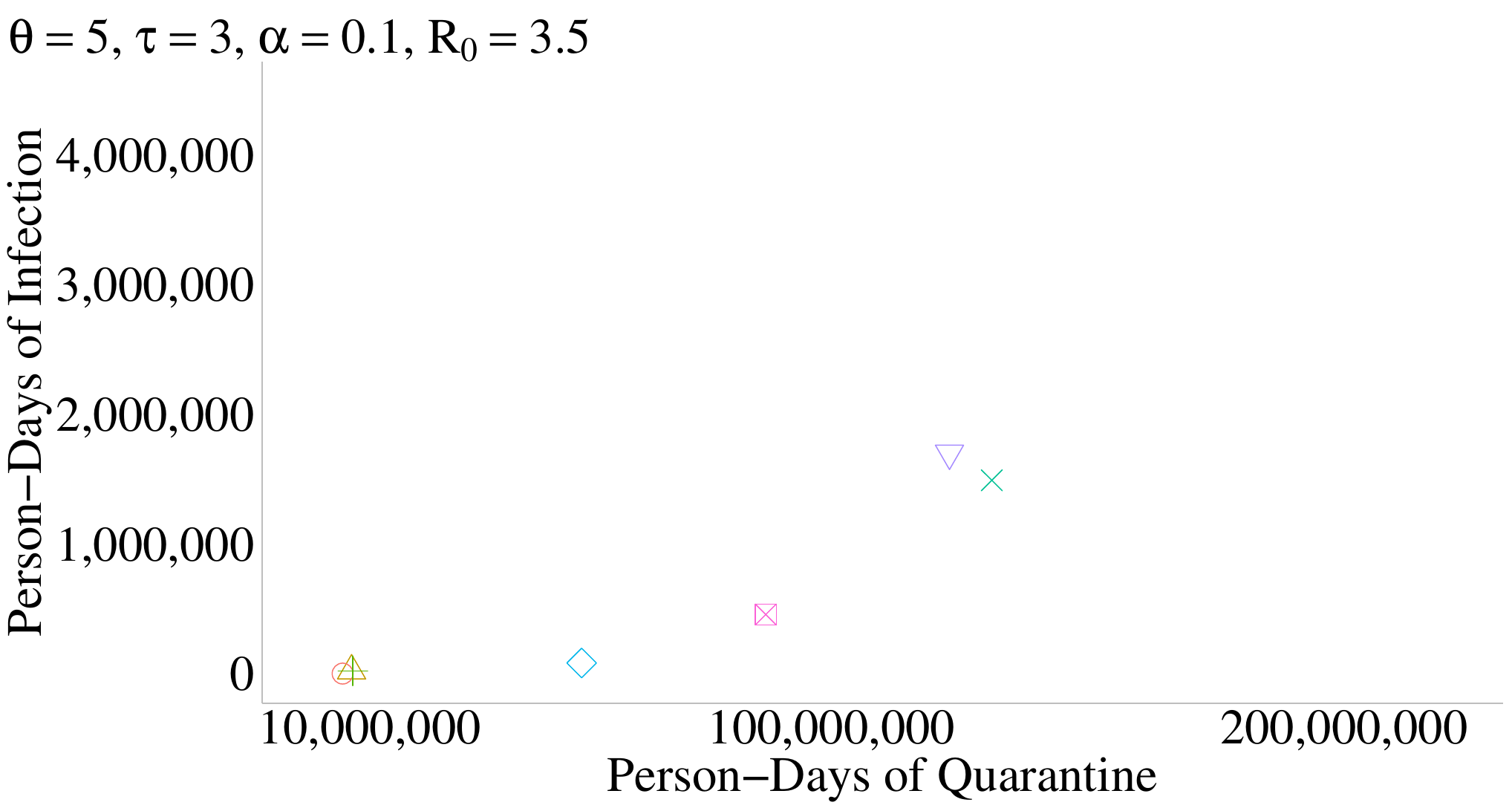}
    }
    \hfill
\subfloat[]{
  \includegraphics[width=0.45\linewidth]{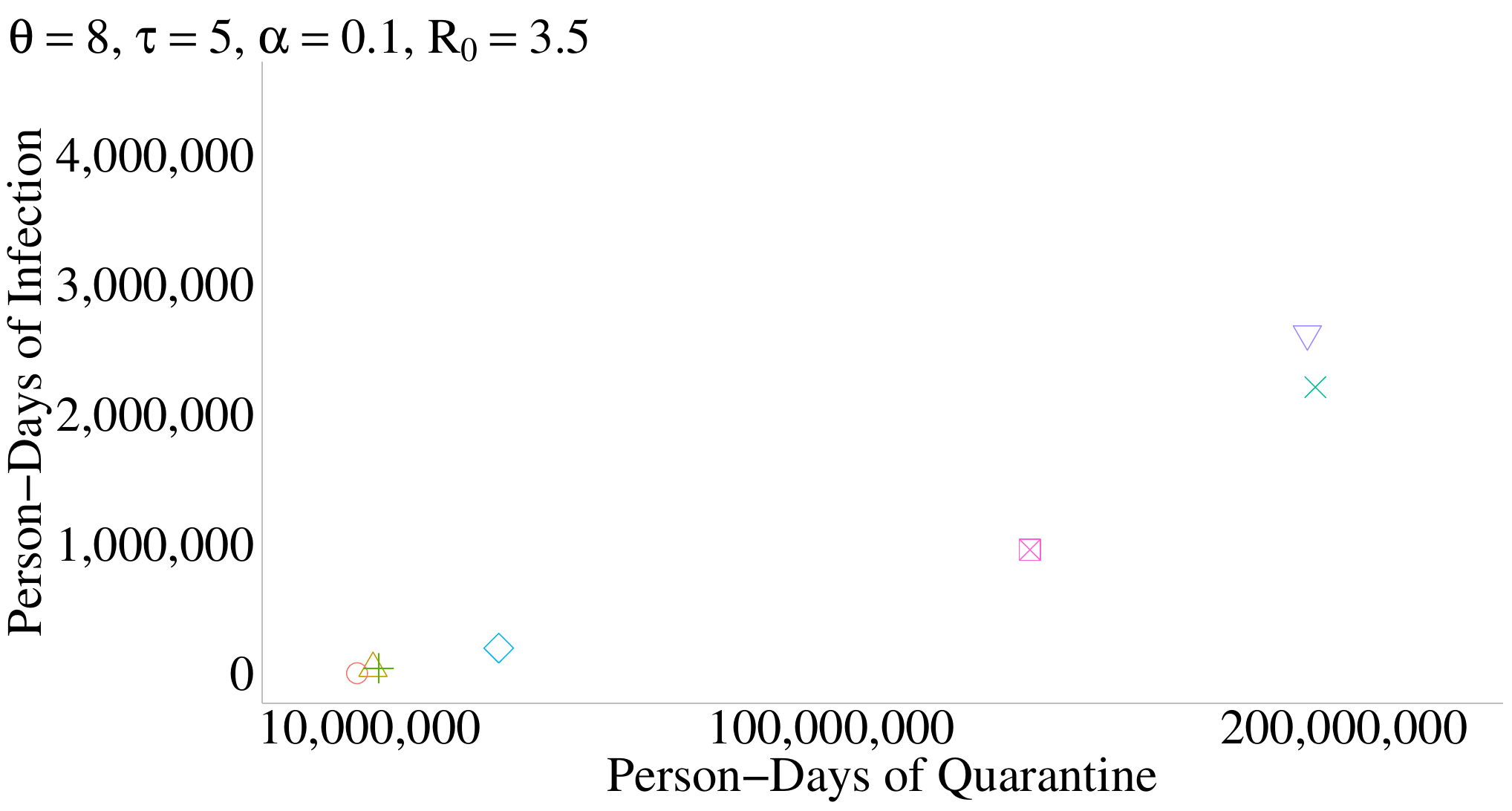}
  } \\
 \subfloat[]{
  \includegraphics[width = 0.45\linewidth]{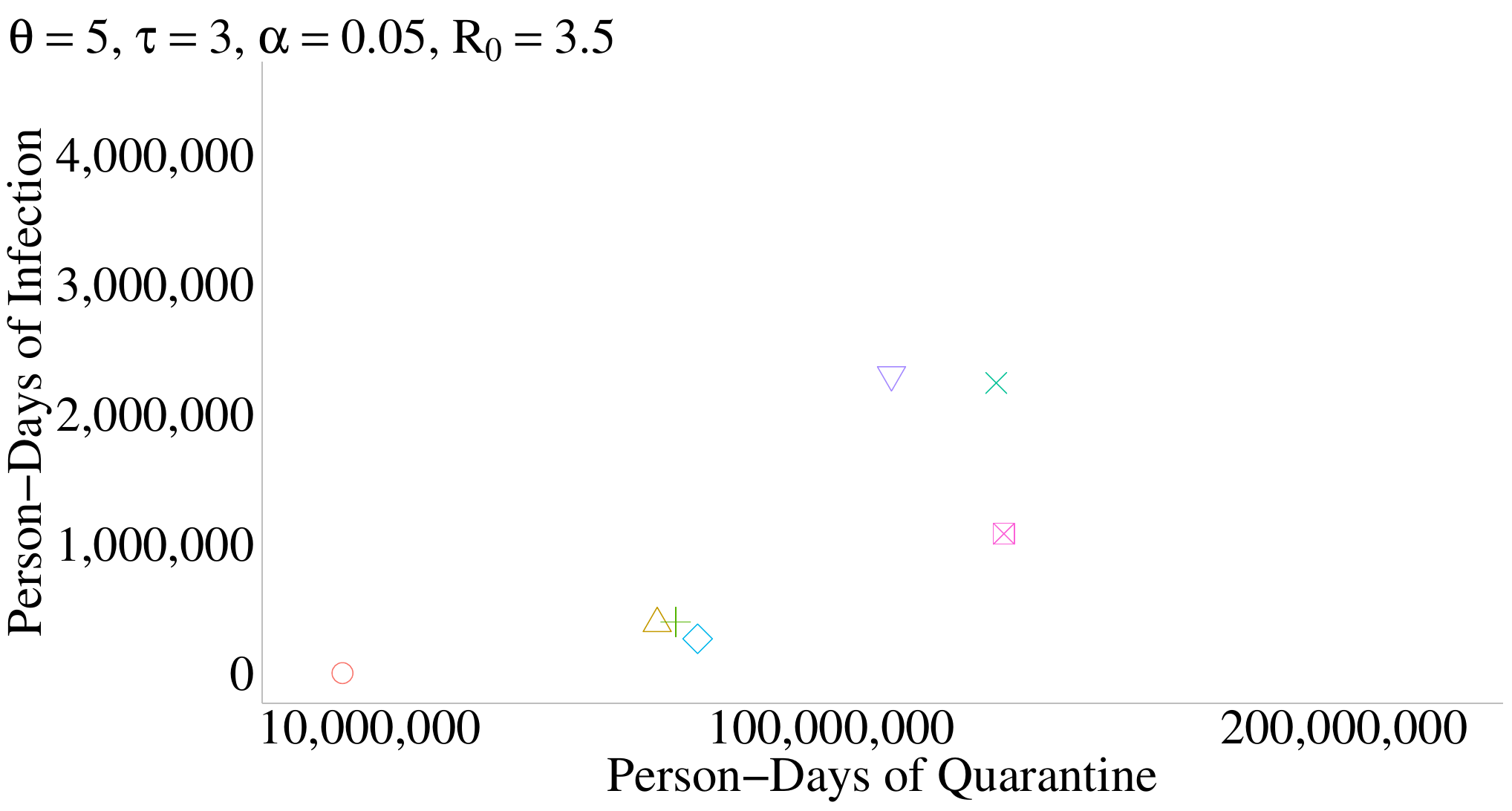}}
\hfill
\subfloat[]{
  \includegraphics[width=0.45\linewidth]{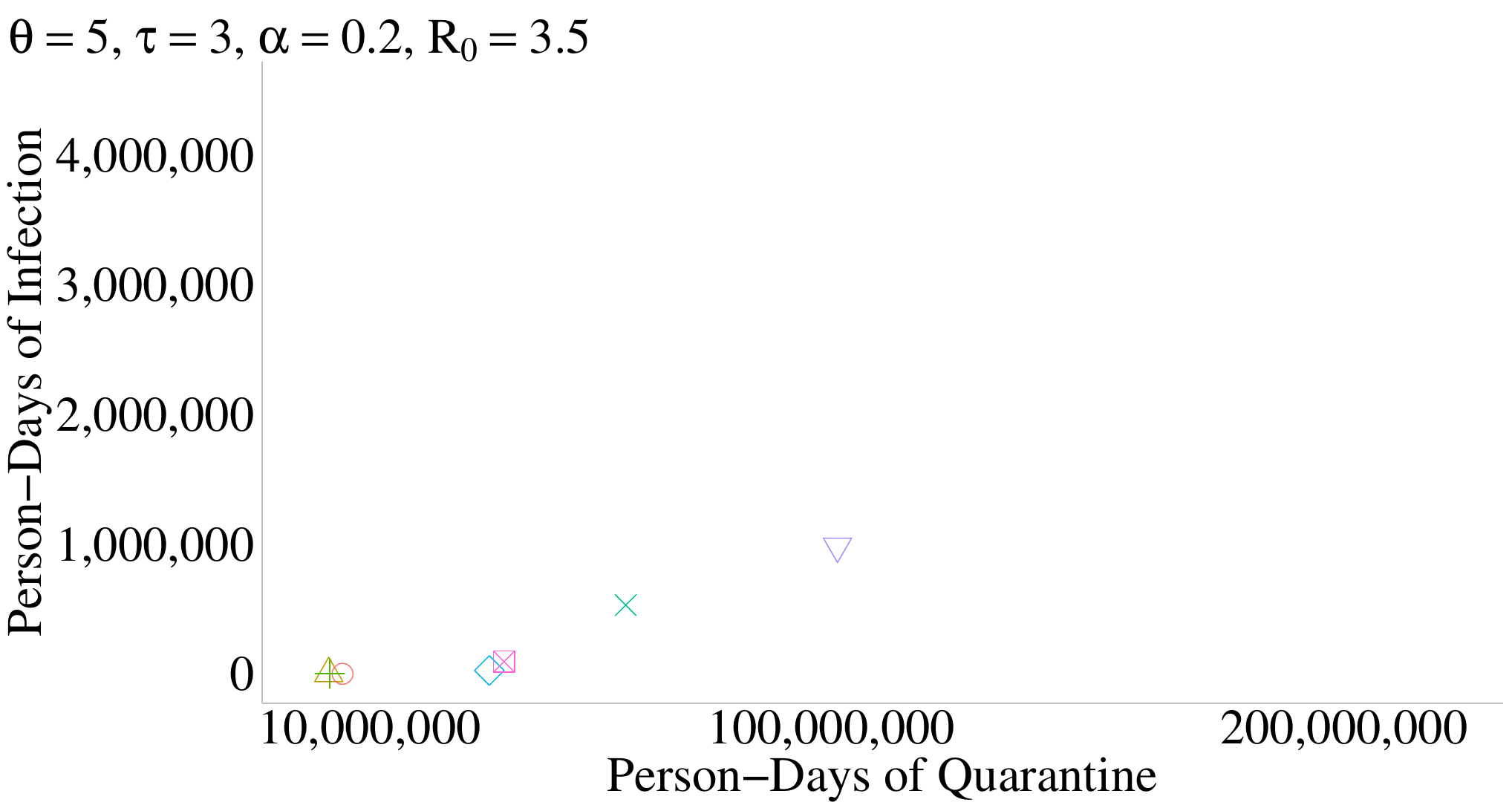}}
\\ \medskip
\subfloat[]{
  \includegraphics[width=0.45\linewidth]{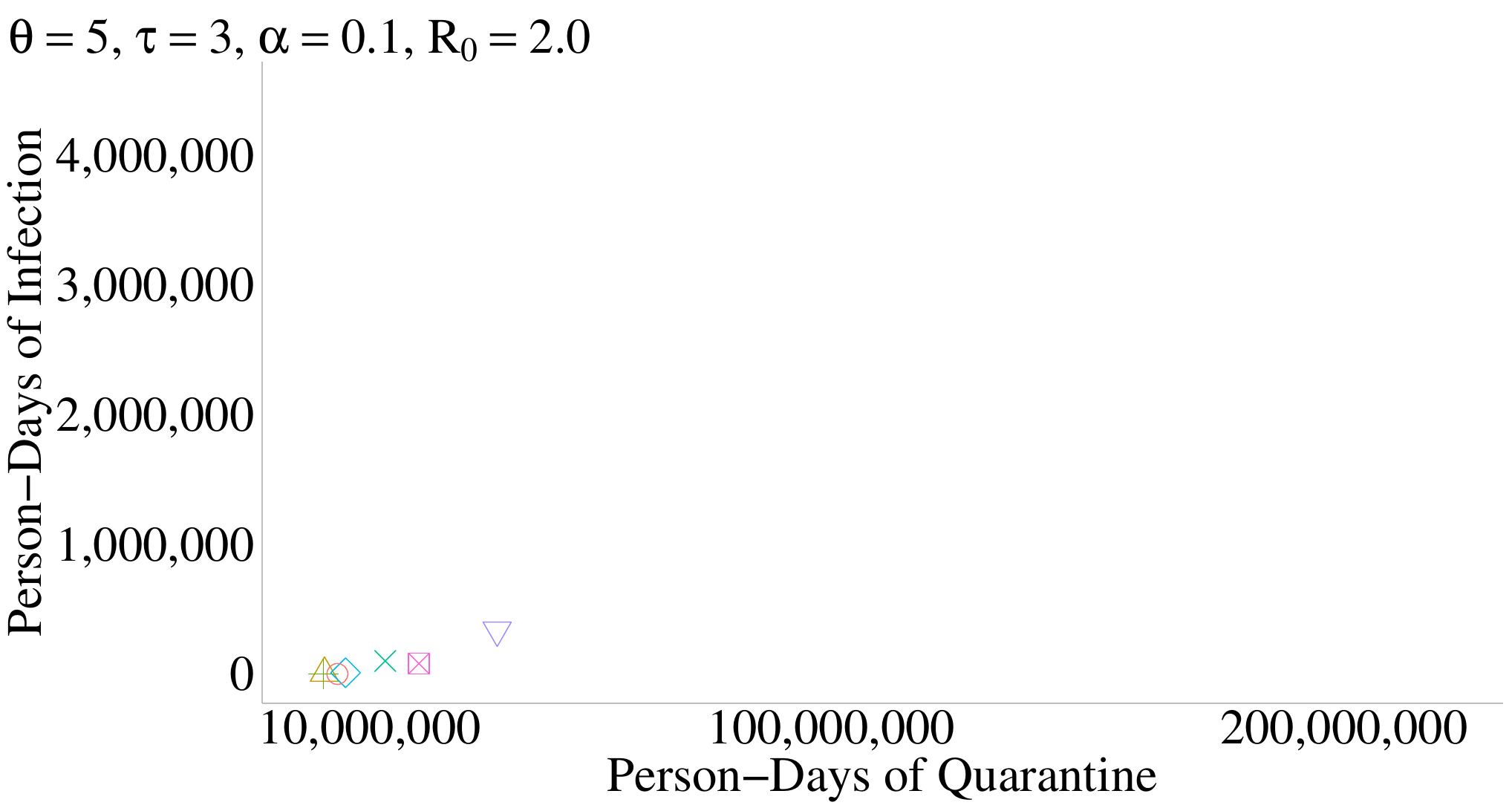}}
\hfil 
\subfloat[]{
  \includegraphics[width=0.45\linewidth]{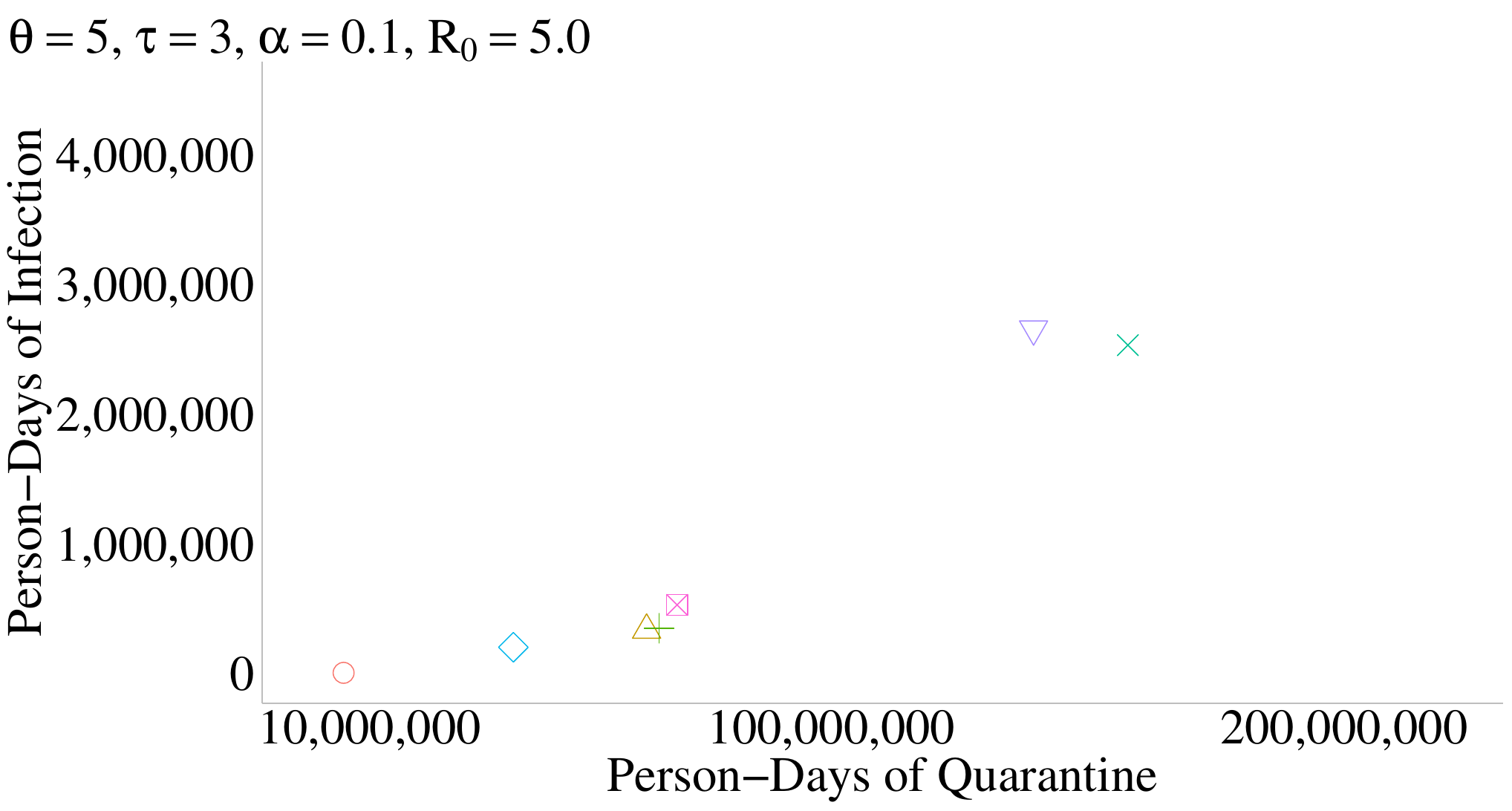}}
\\ \medskip
\subfloat[]{
  \includegraphics[width=0.45\linewidth]{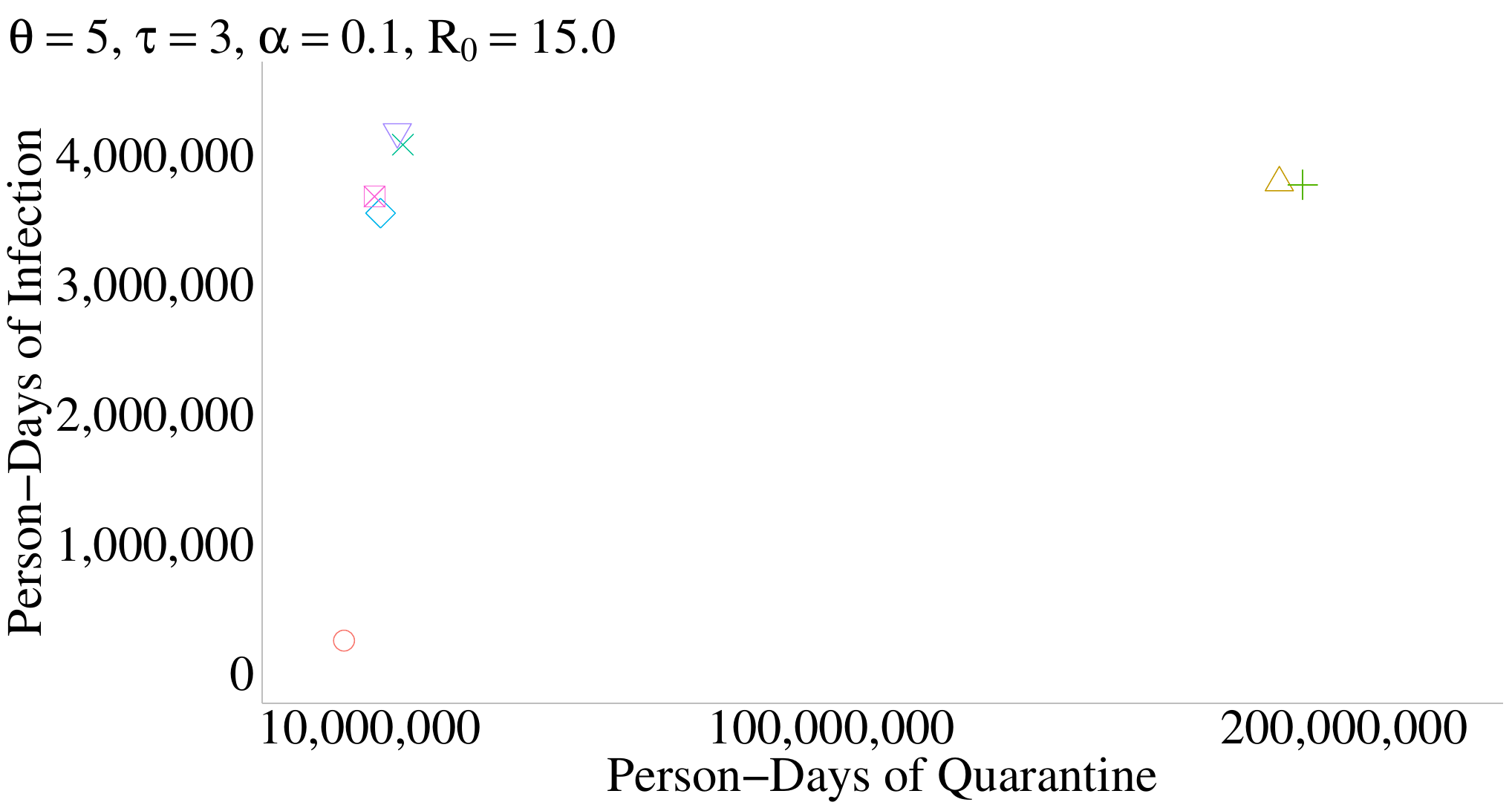}}
 \hfill
\subfloat{
 \includegraphics[width=0.45\linewidth]{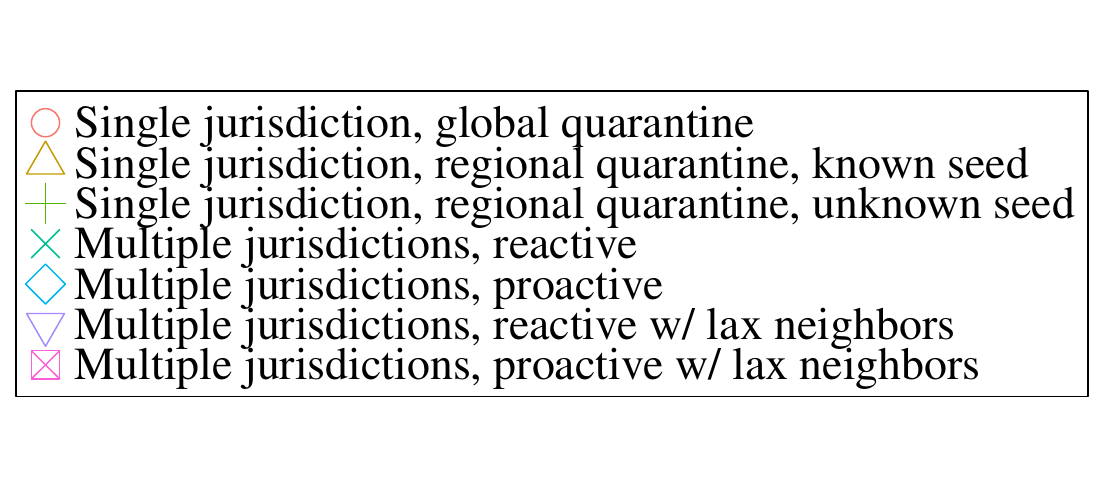}
}
\caption*{\label{fig:images}This figure plots the number of person-days of quarantine (per million) and the number of person-days of infection (per million) for seven different policy scenarios. Both single jurisdiction policies include leakage. Simulations in (a) are the average from 10,000 simulations, while in (b)-(f), the results are the average over 2,500 simulations.}

\end{figure}

\newpage

\end{document}